\documentclass[aos,preprint]{imsart}

\RequirePackage[OT1]{fontenc}
\RequirePackage{amsthm,amsmath}
\RequirePackage[authoryear]{natbib}
\RequirePackage[colorlinks,citecolor=blue,urlcolor=blue]{hyperref}
\usepackage{amssymb}
\usepackage{amsmath}
\arxiv{arXiv:0000.0000}

\startlocaldefs
\numberwithin{equation}{section}
\theoremstyle{plain}

\endlocaldefs
\allowdisplaybreaks 
\newtheorem{theorem}{Theorem}

\newtheorem{assumption}{Assumption}

\newtheorem{corollary}[theorem]{Corollary}

\newtheorem{lemma}{Lemma}

\newcommand{\nc}{\newcommand}
\nc{\norm}{\mathcal{N}}
\nc{\E}{\mathbb{E}}
\nc{\bx}{{\bf x}}
\nc{\bX}{{\bf X}}
\nc{\by}{{\bf y}}
\nc{\IG}{\mathcal{IG}}
\nc{\dd}[2]{\frac{\partial #1}{\partial #2}}
\nc{\lhat}[1][i]{\hat\lambda_{#1}^{-1(g)}}
\nc{\what}[1][j]{\hat\omega_{#1}^{-1(g)}}
\nc{\bone}{{\bf 1}}
\nc{\Li}{\hat\Lambda^{-1(g)}}
\nc{\Oi}{\hat\Omega^{-1(g)}}
\nc{\dps}{\displaystyle}
\nc{\tr}{\text{tr}}
\DeclareMathOperator*{\argmax}{arg\,max}
\nc{\indep}{\mbox{$\perp\!\!\!\perp$}}

\begin{document}

\begin{frontmatter}
\title{Efficient nonparametric estimation of causal mediation effects}
\runtitle{Nonparametric estimation of causal mediation effects}

\begin{aug}
\author{\fnms{K.C.G. Chan}\thanksref{t1}\ead[label=e1]{kcgchan@uw.edu}},
\author{\fnms{K. Imai}\thanksref{t2}\ead[label=e2]{kimai@princeton.edu}},
\author{\fnms{S.C.P. Yam}\thanksref{t3}\ead[label=e3]{scpyam@sta.cuhk.edu.hk}},
\and
\author{\fnms{Z. Zhang}\thanksref{t4}
\ead[label=e4]{zzhang1989@gmail.com}
}

\thankstext{T1}{K.C.G. Chan thanks the United States National Institutes of Health for support (R01 HL122212, R01 AI121259) ; Kosuke Imai thanks the United States National Science
  Foundation for support (SES--0918968); {\color{black}Phillip} Yam acknowledges the
  financial support from the Hong Kong RGC {\color{black}GRF 14301015}, and Direct Grant for Research 2014/15 with project
  code: 4053141 offered by CUHK; Zheng Zhang acknowledges the
  financial support from the Chinese University of Hong Kong and the University of Hong Kong; the present work constitutes part of his research study leading to his Ph.D thesis in the Chinese University of Hong Kong.}  \runauthor{K.C.G. Chan, K. Imai, S.C.P. Yam, and
  Z. Zhang}

\affiliation{University of Washington\thanksmark{t1}, Princeton University\thanksmark{t2},\\ The Chinese University of Hong Kong\thanksmark{t3}, The University of Hong Kong\thanksmark{t4}}

\address{K.C.G. Chan\\
Department of Biostatistics\\
University of Washington\\
\printead{e1}}

\address{K. Imai\\
Department of Politics\\
Center for Statistics and Machine Learning\\
Princeton University\\
\printead{e2}}

\address{S.C.P. Yam\\
Department of Statistics\\
The Chinese University of Hong Kong\\
\printead{e3}
}
\address{Z. Zhang\\
Department of Statistics and Actuarial Science\\
The University of Hong Kong\\
\printead{e4}
}
\end{aug}

\begin{abstract}
  An essential goal of program evaluation and scientific research
  is the investigation of causal mechanisms. Over the past several
  decades, causal mediation analysis has been used in medical and
  social sciences to decompose the treatment effect into the natural
  direct and indirect effects. However, all of the existing mediation
  analysis methods rely on parametric modeling assumptions in one way
  or another, typically requiring researchers to specify multiple
  regression models involving the treatment, mediator, outcome, and
  pre-treatment confounders. To overcome this limitation, we propose a
  novel nonparametric estimation method for causal mediation analysis
  that eliminates the need for applied researchers to model multiple
  conditional distributions. The proposed method balances a certain
  set of empirical moments between the treatment and control groups by
  weighting each observation; {\color{black}in particular, we establish} that the proposed estimator is
  {\it globally} semiparametric efficient. We also show how to
  consistently estimate the asymptotic variance of the proposed
  estimator without additional efforts. Finally, we extend the
  proposed method to other relevant settings including the causal
  mediation analysis with multiple mediators. 
\end{abstract}



\end{frontmatter}

\section{Introduction}

In program evaluation and scientific research, an essential goal is to
understand why and how a treatment variable influences the outcomes of
interest, going beyond the estimation of the average treatment
effects.  In this regard, causal mediation analysis plays an important
role in the investigation of causal mechanisms by decomposing the
treatment effect into the natural direct and indirect effects
\citep{robins1992identifiability,pearl2001direct,robi:03}.  Such an
approach has been widely used in a number of disciplines in medical
and social sciences \citep[see e.g.,][]{baron1986moderator,imai:etal:11,mackinnon2008introduction,vand:15}.
The methodological literature on causal mediation analysis has also
rapidly grown over the last decade and produced numerous approaches
and extensions \citep[see, {\color{black}for example},][]{albe:08,gene:07,imai:keel:yama:10,jo:08,joffe2007defining,sobe:08,ten2007causal,tchetgen2012semiparametric,vanderweele2009marginal,vanderweele2010odds}.

In this {\color{black}article}, we {\color{black} here} contribute to this fast growing literature by
developing a new efficient nonparametric estimation method for causal
mediation analysis.  All of the existing mediation analysis methods
rely on parametric modeling assumptions in one way or another,
typically requiring researchers to specify multiple regression models
involving the treatment $T$, mediator $M$, outcome {\color{black} $Y$,} and
pre-treatment confounders $X$.  For example, the standard approach
based on the so {\color{black} called ``mediation formula''} require the
specification of two or three conditional distributions, {\color{black}i.e.}
$f_{Y \mid M,T,X}$, $f_{M \mid T, X}$, and possibly $f_{T \mid X}$
\citep[{\color{black}for example},][]{imai:keel:ting:10,pear:12,vanderweele2009marginal}.
Inference under this standard approach is only valid when both the
outcome model $f_{Y \mid M, T, X}$ and the mediator model
$f_{M \mid T, X}$ are correctly specified.  Our proposed method
eliminates the need for applied researchers to model these multiple
conditional distributions {\color{black}a priori}.

We {\color{black}are inspired by} the recent work of \cite{tchetgen2012semiparametric} who
develop a robust semiparametric estimation procedure to allow for
possible model misspecification.  The authors show that their proposed
estimator is consistent when any two out of three chosen models are
correctly specified and is locally semiparametric efficient {\color{black}whenever} all
three models are correct.  While this estimator represents an
important advance in the literature, its validity still relies upon
the correct specification of multiple parametric or semiparametric
models.  We improve this estimator by proposing a {\it globally}
semiparametric efficient estimator that attains the semiparametric
efficiency bound, derived by \cite{tchetgen2012semiparametric},
without imposing the additional {\color{black}structural} assumptions required for the existing
semiparametric estimator.  To the best of our knowledge, no globally
semiparametric efficient estimator has been proposed in the causal
mediation literature.

{\color{black} Our} proposed estimator is based on a strategy of balancing covariates
by weighting each observation, which has recently become popular when
estimating the average treatment effects
\citep[{\color{black}for example},][]{chan2015lobally,hain:12,grah:pint:egel:12,imai:ratk:14}.
We combine this idea with the construction of globally semiparametric
efficient estimation of the average treatment effects
\citep[{\color{black}see},][]{chen2008semiprametric,hahn1998role,hirano2000efficient,imbens2005mean}.
Unlike these plugin-type globally semiparametric efficient estimators
that require {\color{black}semi-parametric} estimation of the propensity score or the
outcome regression function, we adopt the nonparametric calibration approach
developed by \cite{chan2015lobally} that {\color{black}constructs}
observation-specific weights only from covariate balancing conditions.
This is a significant advantage in causal mediation analysis because
the plugin-type globally semiparametric estimators would require the
{\color{black}semi-parametric} estimation of three conditional distributions, {\color{black} which is} a
difficult task in practice {\color{black} yielding more doubt on the robustness of the estimators}.

The rest of the paper is organized as follows.  In
Section~\ref{sec:estimation}, we describe the proposed estimation
method, which matches the certain moment conditions of the mediator
and pre-treatment covariates between the treatment and control groups.
We then show how to consistently estimate the asymptotic variance of
the proposed estimator without additional functional estimation.  In
Section~\ref{sec:extensions}, we extend our method to the case of
multiple mediators studied in \citet{imai:yama:13}.  We then discuss
two related estimation problems, namely the estimation of pure
indirect effects and natural direct effect of the untreated.  Finally,
we apply the proposed methods to two data sets in
Section~\ref{sec:dataanalysis} and offer concluding remarks in
Section~\ref{sec:conclusion}.

\section{The Proposed Methodology}
\label{sec:estimation}

In this section, we first consider the efficient nonparametric
estimation of the average natural direct and indirect effects.  In
Theorem~\ref{theorem:main}, we {\color{black}shall} show that the proposed nonparametric
estimator is consistent, asymptotically normal, and globally
semiparametric efficient.  We then demonstrate how to
nonparametrically estimate the asymptotic variance of the proposed
estimator.

\subsection{The framework}

Suppose that we have a binary treatment variable $T \in\{0,1\}$.
Under the standard framework of causal inference, we let $M(t)$ denote
a potential mediating variable, which represents the value of the
mediator if the treatment variable is equal to $t \in \{0,1\}$.
Similarly, let $Y(t,m)$ represent the potential outcome variable under
the scenario where the treatment and mediator variables take the value
$t$ and $m$, respectively.  Then, the observed mediator $M$ is given
by $M = T M(1)+(1-T) M(0)$ whereas the observed outcome is equal to
$Y=T Y(1,M(1))+(1-T) Y(0,M(0))$.  We assume that we have a simple
random sample of size $N$ from a population and therefore observe the
{\it i.i.d.}  realizations of these random variables,
$\{T_i,M_i,Y_i,X_i\}_{i=1}^N$ where $X$ is a vector of pretreatment
covariates.

A primary goal of causal mediation analysis is the following
decomposition of the average treatment effect into the average natural
indirect effect (or average causal mediation effect) and the average
natural direct effect
\citep{robins1992identifiability,pearl2001direct,robi:03}
\begin{eqnarray}
 & & \E[Y(1,M(1))-Y(0,M(0))] \notag \\
& = & \E[Y(1,M(1))-Y(1,M(0))]  + \E[Y(1,M(0))-Y(0,M(0))] \label{eq:ACME}
\end{eqnarray}
The average natural {\color{black}indirect} effect, which is the first term in this
equation, is the average difference between the potential outcome
under the treatment condition and the counterfactual outcome under the
treatment condition where the mediator is equal to the value that
would {\color{black}have realized} under the control condition.  This quantity represents
the average difference that would result if the mediator value changes
from $M(0)$ to $M(1)$ while holding the treatment variable constant at
$T = 1$.  In contrast, the average natural direct effect, which is the
second term in equation~\eqref{eq:ACME}, represents the average
treatment effect when the mediator is held constant at $M(0)$.
Therefore, this decomposition enables researchers to explore how much
of the treatment effect is due to the change in the mediator.  

Note that the following alternative decomposition for causal mediation
analysis is also possible,
\begin{eqnarray}
 & &  \E[Y(1,M(1))-Y(0,M(0))] \notag \\
 & = & \E[Y(0,M(1))-Y(0,M(0))]  + \E[Y(1,M(1))-Y(0,M(1))] \label{eq:ACME2}
\end{eqnarray}
where the treatment variable is held constant at $T=0$ for the natural
indirect effect and the mediator is fixed at {\color{black}$M(1)$} for the natural
direct effect.  \cite{robi:03} called this version of the natural
indirect effect as the {\it pure indirect effect} while referring the
natural indirect effect given in equation~\eqref{eq:ACME} as the {\it
  total indirect effect} {\color{black}since it is resulted from both treatment and mediator}.  {\color{black}Our} proposed estimator is applicable to
both cases as the difference between the two decompositions solely
depends on the value at which the treatment is fixed.

To nonparametrically identify the average natural direct and indirect
effects, we rely on the following set of assumptions as in
\cite{imai:keel:ting:10} and \cite{tchetgen2012semiparametric}.
\begin{assumption}  \label{assump:SI} \
\begin{enumerate}
\item (Consistency) If $T=t$, then $M=M(t)$ with probability 1 for
  $t \in \{0, 1\}$.  If $T=t$ and $M=m$, then $Y=Y(t,m)$ with
  probability 1 for $t \in \{0, 1\}$ and $m \in \mathcal{M}$, where
  $\mathcal{M}$ is the support of the distribution of $M$.

\item (Sequential Ignorability) Given $X$, $\{Y(t^\prime,m),M(t)\}$ is
  independent of $T$ for $t,t^\prime \in \{0, 1\}$.  Also, given $T=t$
  and $X$, $Y(t^\prime ,m)$ is independent of $M(t)$ for
  $t, t^\prime \in \{0, 1\}$ and $m \in \mathcal{M}$.

\item (Positivity) With probability 1 with respect to any $(t, x)$
  where $t \in \{0, 1\}$ and $x \in \mathcal{X}$,
  $f_{M\mid T,X}(m\mid t,x)>0$ for all $m\in \mathcal{M}$ where
  $\mathcal{X}$ is the support of $X$. With probability 1 with respect
  to any $x \in \mathcal{X}$, $f_{T\mid X}(t\mid x)>0$ for all
  $t \in \{0,1\}$.
\end{enumerate}
\end{assumption}

The sequential ignorability assumption is a natural extension of the
unconfoundedness assumption for the identification of the average
treatment effect except that it requires the ``cross-world''
independence between $Y(t^\prime, m)$ and $M(t)$ \citep[\color{black}see,][]{rich:robi:13}.
Several researchers have proposed different sensitivity analysis
techniques for estimating the bias that arises when this assumption is
violated \citep[\color{black} see,][]{imai:keel:yama:10,vand:10,tchetgen2012semiparametric}.
Under Assumption~\ref{assump:SI}, \citet{imai:keel:yama:10} showed
that the average natural direct and indirect effects are
nonparametrically identified.  That is,
\begin{eqnarray*}
\theta_t & =  & \E(Y(1-t,M(t))) \\
& = & \int\int\E(Y\mid T=1-t,M=m,X=x)f_{M\mid T,X}(m\mid T=t,X=x)f_X(x)dxdm
\end{eqnarray*}
and
\begin{eqnarray*}
\delta_t \ = \ \E(Y(t,M(t))) & = & \int \E(Y\mid T=t,X=x)f_X(x)dx
\end{eqnarray*}
for $t=0,1$.  

\cite{tchetgen2012semiparametric} made an important theoretical
advance {\color{black} by showing} that under Assumption~\ref{assump:SI} the efficient
influence function of $\theta_t$ is given by,
\begin{eqnarray}\label{E:sefftheta0}
S_{\theta_t} & = &\frac{\mathbf{1}\{T=1-t\}f_{M\mid T,X}(M\mid
                     T=t,X)}{f_{T\mid X}(1-t\mid X)f_{M\mid T,X}(M\mid
                     T=1-t,X)}\{Y-\E(Y\mid X,M,T=1-t)\}\nonumber \\
               &&+\frac{\mathbf{1}\{T=t\}}{f_{T\mid X}(t\mid
                  X)}\{\E(Y\mid X,M,T=1-t)-\eta(1-t,t,X)\}+\eta(1-t,t,X)-\theta_t \notag\\
\end{eqnarray}
where
\begin{eqnarray*}
  \eta(t,t^\prime,X) & = & \int \E(Y\mid X,M=m,T=t)f_{M\mid T,X}(m\mid T=t^\prime,X)dm 
\end{eqnarray*}
for $t,t^\prime \in \{0,1\}$. {\color{black}Hence, the} definition of $\eta$ implies that
$\eta(1,1,X)=\E(Y\mid X,T=1)$ and $\eta(0,0,X)=\E(Y\mid X,T=0)$.

Furthermore, the efficient influence functions of the average natural
direct effect when $t=0$, i.e., $\textsf{NDE}=\theta_0-\delta_0$, and
the average natural indirect effect when $t=1$ (or the average total
indirect effect), i.e., $\textsf{NIE}=\delta_1-\theta_0$, are
$S_{\textsf{NDE}}=S_{\theta_0}-S_{\delta_0}$ and
$S_{\textsf{NIE}}=S_{\delta_1}-S_{\theta_0}$, respectively, where
$S_{\delta_1}$ and $S_{\delta_0}$ are the efficient influence
functions for estimating $\delta_1$ and $\delta_0$.  As shown in
\cite{robi:rotn:zhao:94} and \cite{hahn1998role}, these efficient
influence functions are given by,
\begin{equation}\label{E:seffdelta}
S_{\delta_t} \ = \ \frac{\mathbf{1}\{T=t\}}{f_{T\mid X}(t\mid
  X)}\{Y-\E(Y\mid X,T=t)\}+\E(Y\mid X,T=t)-\delta_t\ ,
\end{equation}
for $t=0,1$.  Similarly, the average natural indirect effect when
$t=0$ (or the pure indirect effect) is given by
$\textsf{PIE} = \theta_1-\delta_0$ and its efficient function is equal
to $S_{\theta_1}-S_{\delta_0}$. Therefore, the efficient {\color{black}estimations} of
the natural direct and indirect effects {\color{black}involve} the efficient
estimation of $\delta_t$ and $\theta_t$ for $t=0,1$.

\subsection{Efficient estimation of $\delta_1$ and $\delta_0$}

Before proposing a globally semiparametric efficient estimator of
$\theta_0$, which is {\color{black}one of the} main {\color{black}contributions} of {\color{black}our} paper, we discuss the
efficient estimation of $\delta_1$ and $\delta_0$, which is required
for the efficient estimation of the natural direct and indirect
effects.  There exists an extensive literature on the globally
efficient estimation of $\delta_0$ and $\delta_1$ in econometrics
\citep[see, {\color{black}for example},][]{chen2008semiprametric,hahn1998role,hirano2000efficient,imbens2005mean}.
However, many of these existing estimators require the {\color{black}semi-parametric}
estimation of propensity score or outcome regression model.  In this
paper, we focus on a globally efficient estimator recently proposed by
\cite{chan2015lobally}, which serves as a building block of our
proposed estimator of $\theta_0$ discussed below.  Unlike the other
estimators, this approach achieves the efficient nonparametric
estimation by balancing covariates through weighting.

Let $p_0(x) \triangleq \frac{1}{N}f_{T\mid X}(1\mid x)^{-1}$ and
$q_0(x) \triangleq \frac{1}{N}f_{T\mid X}(0\mid x)^{-1}$. Under
Assumption~\ref{assump:SI}, for any suitable integrable functions {\color{black}$u(x)$},
the following important moment conditions hold,
\begin{eqnarray}
\delta_1 & = & \E\left(\sum_{i=1}^N T_ip_0(X_i) Y_i\right) \\
\delta_0 & = & \E\left(\sum_{i=1}^N (1-T_i)q_0(X_i) Y_i \right) \\
\E(u(X)) & = & \E\left(\sum_{i=1}^N T_ip_0(X_i) u(X_i)\right) \label{E:bal1}\\
\E(u(X)) & = & \E\left(\sum_{i=1}^N (1-T_i)q_0(X_i) u(X_i)\right) \label{E:bal0} 
\end{eqnarray}
The first two equalities represent the inverse-probability-weighting
(IPW) estimators of the average potential outcomes.  A number of
scholars have exploited the covariate balance conditions in
equations~\eqref{E:bal1}~and~\eqref{E:bal0} in order to estimate the
average treatment effects
\cite[e.g.,][]{chan2014oracle,han2013estimation,imai:ratk:14,grah:pint:egel:12,qin2007empirical}.
These existing estimators are locally semiparametric efficient, {\color{black}yet all of these}
rely on parametric models in one way or another.

Our goal is, however, to develop a globally {\color{black}fully nonparametric} efficient
estimator.  Thus, we utilize the nonparametric estimator proposed by
\cite{chan2015lobally}.  Let $D(v,v^\prime)$ be a distance measure for
{\color{black}$v,v^\prime \in\mathbb{R}$}. That is, we assume {\color{black}that} $D(v, v^\prime)$
is continuously differentiable in $v\in \mathbb{R}$, non-negative, and
strictly convex in $v$ with $D(v,v)=0$.  Based on
equations~\eqref{E:bal1}~and~\eqref{E:bal0}, \cite{chan2015lobally}
{\color{black}constructs} calibration weights by solving the following minimization
problem subject to constraints that are empirical counterparts of
equations~\eqref{E:bal1}~and~\eqref{E:bal0}:
\begin{eqnarray}\label{E:cm1}
\textsf{Minimize} & &  \sum_{i=1}^N T_i D(N p_i,1) \notag \\
 & &  \textsf{subject
  to}~~\sum_{i=1}^N T_i p_i u_K(X_i) \ = \ \frac{1}{N}\sum_{i=1}^N u_K(X_i),
\end{eqnarray}
and
\begin{eqnarray}\label{E:cm2}
\textsf{Minimize} & &  \sum_{i=1}^N (1-T_i) D(Nq_i,1)\notag \\
                  & & \textsf{subject
                      to}~~\sum_{i=1}^N (1-T_i)q_iu_K(X_i) \ = \ \frac{1}{N}\sum_{i=1}^Nu_K(X_i),
\end{eqnarray}
where $u_K$ is a $K(N)$-dimensional function of $X$, {\color{black}whose components form a set of orthonormal polynomials, here} $K(N)$ increases to infinity when $N$ goes to infinity yet with $K(N)=o(N)$. {\color{black}Furthermore, note that all these $u_{K}$'s have to form a basis on $L^{\infty}$ as $K$ goes to infinity}.

Furthermore, to gain computational efficiency for implementation, they
consider the dual problems of
equations~\eqref{E:cm1}~and~\eqref{E:cm2}.  While the primal problems
given in equations~\eqref{E:cm1}~and~\eqref{E:cm2} are convex
separable programming with linear constraints,
\cite{tseng1987relaxation} showed that the dual problems are
unconstrained convex maximization problems, which can be solved by
efficient and stable numerical algorithms.  With slight abuse of
notation, denote $D(v)=D(v,1)$.  For observations with $T_i = 1$, the
dual solution is given {\color{black}by},
\begin{align} \label{eq:calp}
\hat{p}_K(X_i) \ \triangleq \ \frac{1}{N}
  \rho'\left(\hat{\phi}_K^Tu_K(X_i)\right),
\end{align}
where $\rho'$ is the first derivative of the following strictly
concave function,
\begin{equation}\label{E:dual_rel}
  \rho(v) \ = \ D\left((D')^{-1}(-v)\right)+v\cdot (D')^{-1}(-v),
\end{equation}
and $\hat{\phi}_{K} \in \mathbb{R}^K$ maximizes the following
objective function,
\begin{equation}\label{E:hatF}
  \hat{F}_{K}(\phi) \ \triangleq \ \frac{1}{N}\sum_{i=1}^N
  \left\{T_i\rho\left({\phi}^\top u_K(X_i)\right)-
    \phi^\top {u}_K(X_i)\right\}. 
\end{equation}
Similarly, for observations with $T_i=0$,
\begin{align} \label{eq:calq}
\hat{q}_K(X_i) \ \triangleq \ \frac {1}{N} \rho'\left(\hat{\lambda}_K^Tu_K(X_i)\right),
\end{align}
where $\hat{\lambda}_{K} \in \mathbb{R}^K$ maximizes the following
objective function,
\begin{equation}\label{E:hatG}
  \widehat{G}_{K}(\lambda) \ \triangleq \ \frac{1}{N}\sum_{i=1}^N \left\{(1-T_i)\rho\left({\lambda}^{T}u_K(X_i)\right) - \lambda^{T}{u}_K(X_i)\right\}\ .
\end{equation}

According to the first order conditions for the maximizations given in
equations~\eqref{E:hatF}~and~\eqref{E:hatG}, one can easily verify
that the linear constraints in
equations~\eqref{E:cm1}~and~\eqref{E:cm2} are satisfied.  Finally,
\cite{chan2015lobally} proposed the following empirical covariate
balancing estimator for $\delta_1$ and $\delta_0$,
\begin{equation}
  \hat{\delta}_{1K} \ \triangleq \ \sum_{i=1}^N T_i\hat{p}_K(X_i)Y_i
  \quad \text{and} \quad \hat{\delta}_{0K} \ \triangleq \
  \sum_{i=1}^N(1-T_i)\hat{q}_K(X_i)Y_i
\end{equation}
The authors showed that $\hat{\delta}_{1K}$ and $\hat{\delta}_{0K}$
attain the semiparametric efficiency bounds given in
equation~\eqref{E:seffdelta} under mild regularity conditions.

\subsection{Efficient estimation of $\theta_0$ and $\theta_1$}

We begin by considering the efficient estimation of $\theta_0$.  As
explained below, the same approach can be applied to efficiently
estimate $\theta_1$.  The efficient influence function of $\theta_0$
given in equation~\eqref{E:sefftheta0}, involves three sets of
nonparametric functions: $f_{T\mid X}(1\mid X)$,
$f_{M\mid T,X}(M\mid T=t,X)$ for $t=0,1$, and $\E(Y\mid X,M,T=1)$.
While it is possible to construct a globally efficient estimator of
$\theta_0$ by plugging the corresponding nonparametric estimates into
equation~\eqref{E:sefftheta0}, the performance of the resulting
estimator may be poor because it is difficult to estimate the
conditional density of a possibly continuous mediator and
$f_{M\mid T,X}(M\mid 1,X)$ appears in the denominator of the first
term of $S_{\theta_0}$.  The direct nonparametric estimation of
$f_{M\mid T,X}$ usually results in extreme weights and the
corresponding weighting estimator can become unstable.

Our goal is to construct a weighting estimator for $\theta_0$.  Let us
represent $\theta_0$ as a weighted average of $Y$ among the treated,
\begin{eqnarray}
  \theta_0 & = & \mathbb{E}\left[\frac{TY}{f_{T\mid X}(1\mid X)}\cdot
                 \frac{f_{M\mid T,X}(M\mid 0,X)}{f_{M\mid T,X}(M\mid 1,X)}\right]\ .
\end{eqnarray}
Furthermore, define,
\begin{align}\label{def:r_0}
r_0(m,x)\ \triangleq \ \frac{f_{M\mid T,X}(m\mid 0,x)}{Nf_{T\mid
  X}(1\mid x)f_{M\mid T,X}(m\mid 1,x)}.
\end{align}
If $r_0(m,x)$ is a known function, then a natural estimator for
$\theta_0$ is $\tilde{\theta}_0=\sum_{i=1}^N T_ir_0(M_i,X_i)Y_i$,
which converges to $\theta_0$ by the Law of Large Number.  Since
$r_0(m,x)$ is {\color{black}mostly} unknown, we {\color{black}shall} replace it by an estimated weight.  To
construct moment conditions for estimating $r_0(m,x)$, we need to
develop a covariate balancing property extending
equations~\eqref{E:bal1}~and~\eqref{E:bal0}. This result is given {\color{black}in}
the following lemma.

\begin{lemma}\label{lemma1}
  With $q_0(x)=(Nf_{T\mid X}(0\mid x))^{-1}$ and $r_0(x,m)$ defined in
  equation~\eqref{def:r_0}, we have
 $$\mathbb{E}[T r_0(X,M)v(X,M)]\ = \ \mathbb{E}[(1-T)q_0(X)v(X,M)]$$ 
 for any suitable integrable functions $ v.$
\end{lemma}
\begin{proof}
\begin{align*}
&~~~~\mathbb{E}[T r_0(X, M)v(X, M)]\\
&= \ \mathbb{E}[r_0(X, M)v(X,M)f_{T\mid X,M}(1\mid X,M)]\\
&= \ \mathbb{E}\left[\frac{1}{Nf_{T\mid X}(1\mid X)}\frac{f_{M\mid
  T,X}(M\mid 0,X)}{f_{M\mid T, X}(M\mid 1,X)}v(X, M)f_{T\mid
  X,M}(1\mid X, M)\right]\\
&= \ \int_{\mathcal{M}}\int_{\mathcal{X}}\frac{1}{Nf_{T\mid X}(1\mid
  x)}\frac{f_{M\mid T,X}(m\mid 0,x)}{f_{M\mid T,X}(m\mid
  1,x)}v(x,m)f_{T\mid X,M}(1\mid x,m)f_{X,M}(x,m)dxdm\\
&= \ \int_{\mathcal{M}}\int_{\mathcal{X}}\frac{1}{Nf_{T\mid X}(1\mid
  x)}\frac{f_{M\mid T,X}(m\mid 0,x)}{f_{M\mid T,X}(m\mid 1,x)}v(x,m)f_{T, X, M}(1, x, m)dxdm\\
&= \ \int_{\mathcal{M}}\int_{\mathcal{X}}\frac{f_{M\mid T, X}(m\mid
  0,x)}{Nf_{M,T\mid X}(m, 1\mid x)}v(x,m)f_{T, M\mid X}(1, m\mid x)f_{X}(x)dxdm\\
&= \ \frac{1}{N}\int_{\mathcal{M}}\int_{\mathcal{X}}f_{M\mid T, X}(m|0,x)v(x,m)f_{X}(x)dxdm\\
&= \ \int_{\mathcal{M}}\int_{\mathcal{X}}\frac{v(x,m)}{Nf_{T\mid
  X}(0\mid x)}f_{T\mid X, M}(0\mid x, m)f_{X,M}(x,m)dxdm\\
&= \ \mathbb{E}[q_0(X)v(X,M)f_{T\mid X,M}(0\mid X,M)]\\
&= \ \mathbb{E}[(1-T)q_0(X)v(X,M)].
\end{align*}
\end{proof}

Lemma~\ref{lemma1} motivates us to consider the empirical covariate
balancing weights $\hat{r}_K(x,m)$ {\color{black}which solves for} the following {\color{black}constrained}
optimization problem:
\begin{eqnarray}
 \textsf{Minimize} \quad \sum_{i=1}^N T_iD(Nr_i,1) & &  \textsf{subject to} \label{E:cm3}\\
 \sum_{i=1}^NT_ir_iv_K(X_i,M_i)& = & \sum_{i=1}^N(1-T_i)\hat{q}_K(X_i)v_K(X_i,M_i)\ .  \notag
\end{eqnarray}
Note that $\hat{q}_K(x)$ is constructed from equations~\eqref{eq:calq}~and~\eqref{E:hatG}, and $v_K(x,m)$ with
$K\in\mathbb{N}$ is {\color{black}a} $L$-dimensional {\color{black}vector-valued} function, {\color{black}whose components form a set of orthonormal polynomials}, where
$L=\mathcal{O}(K)$,  {\color{black}i.e.} $L$ is of the same order as $K$. {\color{black} Furthermore, note that these $v_K$'s have to form a basis on $L^{\infty}$ as $K$ goes to infinity.}

The weights {\color{black} $\hat{r}_K(x,m)'$s} are obtained by minimizing the aggregate
distance between the final weights to a vector of constant working
design weights, subject to an empirical analogue of the moment
conditions given in Lemma~\ref{lemma1}.  Unlike {\color{black}to} the case of
\cite{deville1992calibration} who used the true and known design
weights commonly available in sample surveys, the true design weights
for our problem $r_0(m,x)$ is unknown and is a function of two
unknowns $f_{T\mid X}(1\mid X)$ and $f_{M\mid T,X}(M\mid T,X)$.
Therefore, calibration {\color{black}of} true design weights is impossible in this
case.  While the true weights are unavailable, we choose the uniform
working design weights because they make it less likely to yield
extreme weights.  Even with misspecified design weights, however, we
can still show that the proposed weighting estimator is globally
semiparametric efficient.

Similar to equations~\eqref{eq:calp}~and~\eqref{eq:calq}, we can
derive the dual solution for equation~\eqref{E:cm3}. For observations
in the treatment group, i.e., $T_i=1$,
\begin{eqnarray}
\hat{r}_K\left(X_i,M_i\right) & \triangleq & \frac{1}{N}\rho'\left(\hat{\beta}^\top_Kv_K(X_i,M_i)\right),
\end{eqnarray}
where $\rho'$ is the first derivative of the function given in
equation~\eqref{E:dual_rel}, and $\hat{\beta}_K$ maximizes the
following objective function:
\begin{align}\label{eq:Hhat}
\widehat{H}_K(\beta)\ \triangleq\
  \frac{1}{N}\sum_{i=1}^N\left\{T_i\rho\left(\beta^\top
  v_K(X_i,M_i)\right)-N(1-T_i)\hat{q}_K(X_i)\beta^\top v_K(X_i,M_i)\right\} \ .
\end{align}
From  the first order condition of the maximization of $\widehat{H}_K$, we can check that
\begin{equation}
\widehat{H}'_K(\hat{\beta}_K) \ = \
\frac{1}{N}\sum_{i=1}^NT_i\rho\left(\hat{\beta}^\top_Kv_K(X_i,M_i)\right)v_K(X_i,M_i)-\sum_{i=1}^N(1-T_i)\hat{q}_K(X_i)v_K(X_i,M_i)=0.
\end{equation}

Now, we define the proposed estimator for $\theta_0$ to be
\begin{eqnarray}
\hat{\theta}_{0K} & \triangleq & \sum_{i=1}^NT_i\hat{r}_K(X_i,M_i)Y_i.
\end{eqnarray}
Asymptotic properties of $\hat{\theta}_{0K}$ will be derived in the
next subsection.  We {\color{black}shall} show that $\hat{\theta}_{0K}$ is globally
semiparametric efficient under {\color{black}some} mild regularity conditions.  Therefore,
the proposed estimators $\hat{\delta}_{1K}-\hat{\theta}_{0K}$ and
$\hat{\theta}_{0K}-\hat{\delta}_{0K}$ are globally semiparametric
efficient estimators for the average natural indirect and direct
effects {\color{black}respectively}.

The relationship given in equation~\eqref{E:dual_rel} between
$\rho(v)$ and $D(v)$ is derived in the supplementary materials.  In
the supplementary materials, we also show that the strict convexity of
$D$ is equivalent to the strict concavity of $\rho$.  Since the dual
formulation is equivalent to the primal problem, we {\color{black}shall} express the
proposed estimator in terms of $\rho(v)$ in the rest of the
{\color{black}present paper}.  When $\rho(v)=-\exp(-v)$, the weights are equivalent to
the implied weights of exponential tilting
\citep{kitamura1997information}.  When $\rho(v)=\log (1+v)$, the
weights correspond to empirical likelihood (Qin and Lawless, 1994).
When $\rho(v)=-(1-v)^2/2$, the weights are the implied weights of the
continuous updating estimator of the generalized method of moments
\citep{hansen1996finite}.

The proposed estimator $\hat{\theta}_{0K}$ is constructed in a similar
manner as done for $\hat{\delta}_{1K}$ and $\hat{\delta}_{0K}$ (see
Section 2.2).  However, there are some important differences.  First,
$\hat{\delta}_{1K}$ and $\hat{\delta}_{0K}$ only require balancing
pre-treatment variables $X$, but $\hat{\theta}_{0K}$ requires
balancing both pre-treatment variables $X$ and post-treatment
mediators $M$.  Moreover, in equation~\eqref{eq:Hhat}, $\rho$ appears
explicitly in the first term and implicitly in the second term through
the dependency of $\hat{q}_K(x)$ on $\rho$.  This creates new
challenges to the establishment of theoretical results.  Note that
although we consider the weights estimated through a specific
$\rho(v)$, the functional form of the true weights $r_0(x,m)$ is
unspecified.  It will be shown later that any function $r_0(x,m)$,
satisfying a mild differentiability assumption, can be approximated
arbitrarily well by $\hat{r}_K(x,m)$ uniformly {\color{black}as the sample size increases}, so long as $\rho(v)$
also satisfies a mild regularity condition.

We can apply the same methodology to the estimation of $\theta_1$.
Note that
\begin{equation}
  \theta_1 \ = \ \mathbb{E}\left[\frac{(1-T)Y}{f_{T\mid X}(1 \mid X)}\cdot \frac{f_{M\mid T,X}(M\mid 1,X)}{f_{M\mid T,X}(M\mid 0, X)}\right] 
\end{equation}
Define
\begin{eqnarray*}
  w_0(x,m) & \triangleq & \frac{f_{M\mid T,X}(m\mid 1,x)}{Nf_{T,M\mid X}(0,m\mid x)}
\end{eqnarray*}
For any suitable integrable $v(x,m)$,  we have
\begin{align*}
&\mathbb{E}\left[(1-T)\frac{f_{M| T,X}(M|1,X)}{f_{T,M|X}(0,M\mid X)}v(X,M)\right]\\
=\ &\int_{\mathcal{M}}\int_{\mathcal{X}}\frac{f_{M|T,X}(m|1,x)}{f_{T,M\mid X}(0,m\mid x)}v(x,m)f_{T,M,X}(0,m,x)dxdm\\
=\ &\int_{\mathcal{M}}\int_{\mathcal{X}}\frac{f_{M,T,X}(m,1,x)}{f_{T,X}(1,x)}v(x,m)f_{X}(x)dxdm\\
=\ &\mathbb{E}\left[\frac{T v(X,M)}{f_{T\mid X}(1\mid X)}\right] \ .
\end{align*}
Therefore, we can construct empirical covariate balancing weights
$\hat{w}_K$ from the following {\color{black}constrained} maximization problem:
\begin{equation*}
\textsf{Minimize} ~~\sum_{i=1}^N (1-T_i)D\left(Nw_i,1\right)~~\textsf{subject to} ~
\end{equation*}
\begin{align*}
  \sum_{i=1}^N(1-T_i)w_iv_K(X_i,M_i) \ = \ \sum_{i=1}^{N}T_i\hat{p}_K(X_i)v_K(X_i,M_i) \ .
\end{align*}
Its dual solution is given by,
\begin{align*}
  \hat{w}_K(x,m)\ \triangleq \ \frac{1}{N}\rho^\prime \left(\hat{\gamma}_K^\top v_k(x,m)\right) \ ,
\end{align*}
where $\hat{\gamma}_K$ maximizes {\color{black}the} following objective function,
\begin{align*}
\hat{J}_K(\beta)\ \triangleq \ \frac{1}{N}\sum_{i=1}^N\left[(1-T_i)\rho(\gamma^\top v_K(X_i,M_i))-NT_i\hat{p}_K(X_i)\gamma^\top v_K(X_i,M_i)\right]
\end{align*}
Then, we can {\color{black}now suggest} the estimator of
$\textsf{PIE}=\theta_1 - \delta_0$ {\color{black}by:}
\begin{align}
  \widehat{\textsf{PIE}} \ \triangleq \ \hat{\theta}_{1K}-\hat{\delta}_{0K} \ = \ \sum_{i=1}^N(1-T_i)\hat{w}_K(X_i,M_i)Y_i-\sum_{i=1}^N(1-T_i)\hat{q}_K(X_i)Y_i \ .\label{eq:PIE} 
\end{align}

\subsection{Asymptotic properties} 

To derive the asymptotic properties of the proposed estimator, we list
all additional {\color{black}technical} assumptions that are required beyond
Assumption~\ref{assump:SI}.
\begin{assumption} \label{as:EY2}
  $\mathbb{E}(Y^2\mid T=0) <\infty$ and $\mathbb{E}(Y^2\mid T=1) < \infty.$
\end{assumption}

\begin{assumption} \label{as:distribution} \
  \begin{enumerate}
  \item The support $\mathcal{X}$ of $r_1$-dimensional covariate $X$ is
    a Cartesian product of $r_1$ compact intervals.
  \item The support $\mathcal{M}$ of $r_2$-dimensional mediating
    variable $M$ is a Cartesian product of $r_2$ compact
    intervals. 
  \end{enumerate}
  Denote $r\triangleq r_1+r_2 .$
\end{assumption}

\begin{assumption}\label{as:bdd}
  There exist some constants
  $\eta_1, \eta_2, \eta_3, \eta_4, \eta_5,\eta_6$ such that the
  following inequalities hold:
  \begin{eqnarray*}
    0\ < \ \frac{1}{\eta_1}\ \leq & f_{T\mid X}(0\mid x) & \leq \ \frac{1}{\eta_2} \ <1,\\
    0\ < \ \frac{1}{\eta_3}\ \leq & f_{M\mid T,X}(m\mid 0,x) & \leq \ \frac{1}{\eta_4} \ <1,\\
    0\ < \ \frac{1}{\eta_5}\ \leq & f_{M\mid T,X}(m\mid 1,x) & \leq \ \frac{1}{\eta_6} \ <1.
  \end{eqnarray*}
\end{assumption}

\begin{assumption}\label{as:differentiable} The functions,
  $q(x)$ and $r(x,m)$, are $s$-times and $s^\prime$-times continuously
  differentiable, respectively, where $s^\prime > 19r>0$ and
  $s>16r_1>0$.
\end{assumption}

\begin{assumption} \label{as:Q0} The function
  $\mathbb{E}(Y\mid T=1,M=m,X=x)$ is $t$-times {\color{black}jointly} continuously
  differentiable with respect to $(x,m)$, and $\eta(1,0,x)$ is
  {\color{black}$d^\prime$}-times continuously differentiable {\color{black} w.r.t. $x$}, where {\color{black}$d >3r/2$} and
  {\color{black}$d^\prime > 3r_1/2$}.
\end{assumption}

\begin{assumption} \label{as:KN} $K = O(N^{\nu})$ and
  $\left(\dps \frac {1} {\dps s^\prime/r -2}\vee
    \frac{1}{s/r_1+1}\right) < \nu < \frac {1} {17}.$
\end{assumption}

\begin{assumption} \label{as:rho} $\rho\in C^{3}(\mathbb{R})$ is a
  strictly concave function defined on $\mathbb{R}$ i.e.,
  $\rho''(\gamma) < 0, ~ \forall \gamma \in \mathbb{R}$, and the range
  of $\rho'$ contains {\color{black} the following subset of the positive real line: $$[\eta_2,\eta_1]\cup \left[\frac{\eta_1}{\eta_1-1},\frac{\eta_2}{\eta_2-1}\right]\cup\left[\frac{\eta_1\eta_6}{\left(\eta_1-1\right)\eta_3},\frac{\eta_2\eta_5}{\left(\eta_2-1\right)\eta_4}\right]\ .$$}
\end{assumption}

Assumptions~\ref{assump:SI}--\ref{as:KN} or similar assumptions also
appeared in the literature
\citep[e.g.,][]{chan2015lobally,hahn1998role,hirano2000efficient,imbens2005mean}.
As explained earlier, Assumption~\ref{assump:SI} is used for the
identification of the natural direct and indirect effects.
Assumption~\ref{as:EY2} is required for the finiteness of asymptotic
variance.  Assumptions~\ref{as:distribution}~and~\ref{as:bdd} are
needed to establish the uniform boundedness of approximations.
Assumptions~\ref{as:differentiable}~and~\ref{as:Q0} are required for
controlling the remainder of approximations with a given {\color{black}set of} basis
{\color{black}functions}.  Assumption~\ref{as:rho} is required for controlling the stochastic order of the reminder terms, which is satisfied by commonly used $\rho$ functions as discussed above.  This final assumption imposes
a mild regularity condition on $\rho$.  \cite{chan2015lobally}
{\color{black}maintains} the same assumption.

Two intermediate lemmas are needed to prove the main theorem.  We
define the following intermediate quantities that are probability
limits of
$\widehat{F}_K, \hat{\phi}_K,\hat{p}_K, \widehat{G}_K,
\hat{\lambda}_K,\hat{q}_K, \widehat{H}_k,\hat{\beta}_K$
and $\hat{r}_K$ for {\color{black}each} fixed $K$:
\begin{eqnarray*}
  F^*_K(\phi) & \triangleq &  \mathbb{E}\left[T\rho\left(\phi^\top u_K(X)\right)-\phi^\top u_K(X)\right] \ = \ \mathbb{E}\left(\widehat{F}_K(\phi)\right)\ ,\\
 \phi^*_K & \triangleq & \argmax_{\phi \in \mathbb{R}^K} F^*_K(\phi)\ ,\\
  p_K^*(x) & \triangleq & \frac {1}{N}\rho^\prime\left((\phi_K^*)^\top u_K(x)\right)\ ,\\
 G^*_K(\lambda) & \triangleq & \mathbb{E}\left[(1-T)\rho\left(\lambda^\top u_K(X)\right)-\lambda^\top u_K(X)\right] = \mathbb{E}\left(\widehat{G}_K(\lambda)\right)\ ,\\
\lambda^*_K & \triangleq & \argmax_{\lambda \in \mathbb{R}^K} G^*_K(\lambda)\ ,\\
q_K^*(x) & \triangleq & \frac {1}{N}\rho^\prime\left((\lambda_K^*)^\top u_K(x)\right)\ ,\\
H^*_K(\beta ) & \triangleq & \mathbb{E}\left[T\rho\left(\beta^\top v_K(X,M)\right)-(1-T)(f_{T\mid X}(0\mid X))^{-1}\beta^\top v_K(X,M)\right]\ ,\\
\beta ^*_K & \triangleq & \argmax_{\beta  \in \mathbb{R}^K} H^*_K(\beta )\ ,\\
r_K^*(x,m) & \triangleq & \frac {1}{N}\rho^\prime\left((\beta_K^*)^\top v_K(x,m)\right)\ .
\end{eqnarray*}

Also, let $\zeta(K)=\sup_{x\in \mathcal{X}}||u_K(x)||$.  The following
lemma establishes the approximation of functions $p_0(x)$, $q_0(x)$,
$r_0(x,m)$ by $p_K^*(x)$, $q_K^*(x)$ and $r_K^*(x,m)$.
\begin{lemma}\label{lemma2}
  Under
  Assumptions~\ref{as:distribution},~\ref{as:bdd},~\ref{as:differentiable},~and~\ref{as:KN},
  we have,
  \begin{eqnarray*}
    \sup_{x\in\mathcal{X}}|Np_0(x)-Np_K^*(x)| & = & O\left(K^{-\frac{s}{2r_1}}\zeta(K)\right) \ ,\\
    \sup_{x\in\mathcal{X}}|Nq_0(x)-Nq_K^*(x)| & = &
                                                    O\left(K^{-\frac{s}{2r_1}}\zeta(K)\right) \ ,\\
    \sup_{(x,m)\in\mathcal{X}\times\mathcal{M}}|Nr_0(x,m)-Nr^*_K(x,m)| 
                                              & = &
                                                    O\left(K^{-\frac{s'}{2r}}\zeta(K)\right) \ .
  \end{eqnarray*}
\end{lemma}
\begin{proof}
  The proof is given in the supplementary material.
\end{proof}

The other lemma is about the performance of the estimated auxiliary
parameters, $\hat{\phi}_K$, $\hat{\lambda}_K$, and $\hat{\beta}_K$
that maximize
equations~\eqref{E:hatF},~\eqref{E:hatG},~and~\eqref{eq:Hhat} {\color{black}respectively}.
\begin{lemma} \label{lemma3} Under
  Assumptions~\ref{as:distribution},~\ref{as:bdd},~\ref{as:differentiable},~and~\ref{as:KN},
  we have
\begin{eqnarray*}
  \|\hat{\phi}_K-\phi_K^*\| & = & O_p\left(\sqrt{\frac{K}{N}}\right)\ ,\\
  \|\hat{\lambda}_K-\lambda_K^*\| & = & O_p\left(\sqrt{\frac{K}{N}}\right)\ ,\\
  \|\hat{\beta}_K-\beta_K^*\| & = & O_p\left(\sqrt{\frac{K^5}{N}}\right)\ .
\end{eqnarray*}
\end{lemma}
\begin{proof}
The proof is given in the supplementary material.
\end{proof}

The following theorem shows that $\hat{\theta}_{0K}$ is consistent,
asymptotically normal, and globally semiparametric efficient.
\begin{theorem}\label{theorem:main}
  Under Assumptions~\ref{assump:SI}--\ref{as:rho}, $\hat{\theta}_{0K}$
  has the following properties:
\begin{enumerate}
\item  $\hat{\theta}_{0K}=\dps \sum_{i=1}^N T_i\hat{r}_K(X_i,M_i)Y_i \stackrel{p}{\longrightarrow} \theta_0 $;
\item
  $ \dps \sqrt{N}\left(\sum_{i=1}^N T_i\hat{r}_K(X_i,M_i)Y_i -
    \theta_0 \right) \stackrel {d} {\longrightarrow} \mathcal{N}(0,
  V_{\theta_0})$,
  where $\dps V_{\theta_0} = \mathbb{E}\left(S_{\theta_0}^2\right)$,
  attains the semi-parametric efficiency bound
  \citep{tchetgen2012semiparametric} with the definition of
  $S_{\theta_0}$ {\color{black}as} given in equation~\eqref{E:sefftheta0}.
\end{enumerate}
\end{theorem}
\begin{proof}
The proof is given in the supplementary material.
\end{proof}

The following corollary establishes the large sample properties of the
estimated average natural indirect effect,
$\widehat{\textsf{NIE}}_K=\hat\delta_{1K} - \hat\theta_{0K}$, and {\color{black}of} the
estimated average natural direct effect,
$\widehat{\textsf{NDE}}_K=\hat\theta_{0K} - \hat\delta_{0K}$.
\begin{corollary}
  Under Assumptions~\ref{assump:SI}--\ref{as:rho},
  $\widehat{\textsf{NIE}}_K$ and $\widehat{\textsf{NDE}}_K$ have the
  following properties:
\begin{enumerate}
\item Consistency.
  \begin{eqnarray*}
  \widehat{\textsf{NIE}}_{K} & = & \dps \sum_{i=1}^N T_i\hat{p}_K(X_i)Y_i-
  \sum_{i=1}^N T_i\hat{r}_K(X_i,M_i)Y_i \stackrel{p}{\longrightarrow}
  \textsf{NIE} = \delta_{1K} - \theta_{0K}\ ,\\
  \widehat{\textsf{NDE}}_{K} & = & \dps
  \sum_{i=1}^N(1-T_i)\hat{r}_K(X_i,M_i)Y_i-
  \sum_{i=1}^N(1-T_i)\hat{q}_K(X_i)Y_i \stackrel{p}{\longrightarrow}
  \textsf{NDE} = \theta_{0K} - \delta_{0K} \ .
  \end{eqnarray*}
\item Asymptotic normality and semiparametric
  efficiency. 
\begin{eqnarray*} 
\sqrt{N}\left(\widehat{\textsf{NIE}}_{K}- \textsf{NIE} \right) & \stackrel {d} {\longrightarrow} & \mathcal{N}(0, V_{\textsf{NIE}})\ ,\\
\sqrt{N}\left(\widehat{\textsf{NDE}}_{K}- \textsf{NDE} \right) & \stackrel {d} {\longrightarrow} & \mathcal{N}(0, V_{\textsf{NDE}}) \ ,
\end{eqnarray*}
where
$\dps V_{\textsf{NIE}} = \mathbb{E}\left[
  \left(S_{\delta_1}-S_{\theta_0}\right)^2\right]$
and
$\dps V_{\textsf{NDE}} = \mathbb{E}\left[
  \left(S_{\theta_0}-S_{\delta_0}\right)^2\right]$,
with $S_{\theta_0}$, $S_{\delta_1}$, and $S_{\delta_0}$ {\color{black}as} defined in
equations~\eqref{E:sefftheta0},~and~\eqref{E:seffdelta}, respectively.
\end{enumerate}
\end{corollary}

Similarly, we can derive the analogous asymptotic properties for the
pure indirect effect.
\begin{theorem}
  Under Assumptions~\ref{assump:SI}--\ref{as:rho},
  $\widehat{\textsf{PIE}}$ defined in equation~\eqref{eq:PIE} has the
  following properties:
\begin{enumerate}
\item  $\widehat{\textsf{PIE}}  \stackrel{p}{\longrightarrow}  \textsf{PIE}$
\item
  $\sqrt{N}\left(\widehat{\textsf{PIE}}- \textsf{PIE}\right)
  \stackrel{d}{\longrightarrow} \mathcal{N}(0, V_{\textsf{PIE}})$,
  where
  $V_{\textsf{PIE}} = \mathbb{E}\left[(S_{\textsf{PIE}})^2\right]$
  attains the semi-parametric efficiency bound
  \citep{tchetgen2012semiparametric}.  The definition of
  $S_{\textsf{PIE}}$ is given by,
\begin{align*}
  S_{\textsf{PIE}}\ = \ &\frac{\mathbf{1}\{T=0\}f_{M\mid T,X}(M\mid 1,X)}{f_{T\mid X}(0\mid X)f_{M\mid T,X}(M\mid 0,X)}\left\{Y-\mathbb{E}\left(Y\mid X,M,T=0\right)\right\}\\
                        &+\frac{\mathbf{1}\{T=1\}}{f_{T\mid X}(1\mid X)}\left\{\mathbb{E}\left(Y\mid X,M,T=0\right)-\eta(0,1,X)\right\}\\
                        &-\frac{\mathbf{1}\{T=0\}}{f_{T\mid X}(0\mid X)}\left\{Y-\eta(0,0,X)\right\}+\eta(0,1,X)-\eta(0,0,X)-\textsf{PIE}.
\end{align*}
\end{enumerate}
\end{theorem}

\subsection{Nonparametric estimation of the asymptotic variance}
\label{sec:variance}

We have shown that the proposed estimator attains the semiparametric
efficiency bound whose efficient influence function depends on three
sets of nonparametric functions.  In this section, we propose a
consistent variance estimator that does not require additional
nonparametric function estimates and is easy to compute.  Define the
following quantities:
\begin{eqnarray*}
  \tau & \triangleq & (\phi^\top,\lambda^\top, \beta^\top, \delta,\delta^\prime,\theta)^\top \ ,\\
  g_{K1}(T,X;\phi) & \triangleq & T \rho^\prime\left(\phi^\top u_K(X)\right)u_K(X)-u_K(X)\ ,\\
  g_{K2}(T,X;\lambda) & \triangleq & (1-T)\rho^\prime\left(\lambda^\top u_K(X)\right)u_K(X)-u_K(X)\ ,\\
  g_{K3}(T,X,M;\lambda,\beta) & \triangleq & T\rho^\prime\left(\beta^\top
                                             v_K(X,M)\right)v_K(X,M)-
                                             (1-T)\rho^\prime\left(\lambda^\top
                                             u_K(X)\right)v_K(X,M)\ ,\\
  g_{K4}(T,X,Y;\phi,\delta) & \triangleq & T\rho^\prime\left(\phi^\top u_K(X)\right)Y -\delta\ ,\\
  g_{K5}(T,X,Y;\lambda,\delta) & \triangleq &
                                              (1-T)\rho'\left(\lambda^\top u_K(X)\right)Y - \delta^\prime\ ,\\
  g_{K6}(T,X,M,Y;\beta,\theta) & \triangleq & T\rho^\prime\left(\beta^\top v_K(X,M)\right)Y -\theta\ ,\\
  g_K(T,X,M,Y;\tau) & \triangleq & \left(g_{K1}^\top,g_{K2}^\top,g_{K3}^\top,g_{K4},g_{K5},g_{K6}\right)^\top\ ,\\
\hat{\tau}_K & \triangleq & \left(\hat{\phi}_K^\top,\hat{\lambda}_K^\top, \hat{\beta}_K^\top, \hat{\delta}_{1K},\hat{\delta}_{0K},\hat{\theta}_{0K}\right)^\top\ ,\\
\tau_K^{*} & \triangleq & \left((\phi_K^{*})^\top,(\lambda_K^{*})^\top, (\beta_K^*)^\top, \delta_{1K}^*,\delta_{0K}^*,\theta_{0K}^*\right)^\top\ ,
\end{eqnarray*}
where $\delta,\delta^\prime,\theta\in\mathbb{R}$,
$\delta^*_{1K} \triangleq \mathbb{E}\left[T Np_K^*(X)Y\right]$,
$\delta^*_{0K} \triangleq \mathbb{E}\left[(1-T)Nq_K^*(X)Y\right]$ and
$\theta^*_K \triangleq \mathbb{E}\left[T Nr_K^*(X,M)Y\right]$.  Note
that by definition $\hat{\tau}_K$ satisfies:
\begin{align}\label{eq:constraint}
  \frac{1}{N}\sum_{i=1}^N g_K(T_i,X_i,M_i,Y_i;\hat{\tau}_K) \ = \ 0 \ .
\end{align}
Applying the Taylor's {\color{black}theorem} for \eqref{eq:constraint} at $\tau_K^*$
yields:
\begin{equation}\label{eq:empirical}
0 \ = \ \frac{1}{N}\sum_{i=1}^N g_K(T_i,X_i,M_i,Y_i;\tau^*_K)+ \frac{1}{N}\sum_{i=1}^N\frac{\partial g_K(T_i,X_i,M_i,Y_i;\tilde{\tau}_K)}{\partial \tau}(\hat{\tau}_K-\tau^*_K) \ ,
\end{equation}
where $\tilde{\tau}_K$ lies on the line joining $\hat{\tau}_K$ with
$\tau^*_K.$ In the supplementary material, we show {\color{black}that}:
\begin{align} \label{eq:op1}
\frac{1}{N}\sum_{i=1}^N\frac{\partial
  g_K(T_i,X_i,M_i,Y_i;\tilde{\tau}_K)}{\partial \tau} \ = \ \mathbb{E}\left[\frac{\partial g_K(T,X,M,Y;\tau^*_K)}{\partial \tau}\right]+ o_p(1) \ ,
\end{align}
where
\begin{align*}
&~~~\mathbb{E}\left[\frac{\partial g_K(T,X,M,Y;\tau^*_K )}{\partial \tau}\right]
\ = \ \begin{pmatrix}
A_{3K\times 3K},&B_{3K\times 3}\\
C_{3\times 3K},&D_{3\times 3}
\end{pmatrix} 
\end{align*}
and 
\begin{eqnarray*}
A_{3K\times 3K} & \triangleq &
\begin{pmatrix}
A_{11}&0_{K \times K}&0_{K \times K}\\
0_{K \times K}&A_{22}&0_{K \times K}\\
0_{K \times K}&A_{32}&A_{33}
\end{pmatrix}\ ,\\
B_{3K\times 3} & \triangleq &  0_{3K\times 3}\ ,\\
C_{3\times 3K} & \triangleq &
\begin{pmatrix}
C_{11}&0_{1 \times K}&0_{1 \times K}\\
0_{1 \times K}&C_{22}&0_{1 \times K}\\
0_{1 \times K}&0_{1 \times K}&C_{33}
\end{pmatrix}\ ,\\
D_{3\times 3} & \triangleq & -I_{3 \times 3}\ ,\\
A_{11} & = & \E\left[T\rho''((\phi_K^*)^\top u_K(X))u_K(X)u_K^\top(X)\right]\ ,\\
A_{22} & = & \E[(1-T)\rho''((\lambda_K^*)^\top u_K(X))u_K(X)u_K^\top(X)]\ ,\\
A_{32} & = & -\E[(1-T)\rho''((\lambda_K^*)^\top u_K(X))v_K(X,M)u_K^\top(X)]\ ,\\
A_{33} & = & \E[T\rho''((\beta_K^*)^\top v_K(X,M))v_K(X,M)v_K^\top(X,M)]\ ,\\
C_{11} & = & \E[T\rho''((\phi_K^*)^\top u_K(X))Yu_K^\top (X)]\ ,\\
C_{22} & = & \E[(1-T)\rho''(-(\lambda_K^*)^\top u_K(X))Yu_K^\top(X)]\ ,\\
C_{33} & = & \E[T\rho''((\beta_K^*)^\top v_K(X,M))Yv_K^\top(X,M)]\ .
\end{eqnarray*}
The parameters of interest are  
\begin{equation*}
\begin{pmatrix}
\textsf{NIE} \\ \textsf{NDE} \\ \theta_0
\end{pmatrix} \ = \ \begin{pmatrix}
\delta_1-\theta_0 \\ \theta_0-\delta_0 \\ \theta_0
\end{pmatrix} \  = \ \begin{pmatrix}
1 & 0 & -1 \\ 0 &-1 & 1 \\ 0 & 0 &1
\end{pmatrix}\pi
\end{equation*}
where $\pi\triangleq \left(\delta_1,\delta_0,\theta_0\right)^\top$.
{\color{black}Note} that asymptotic variance of estimators
$\hat{\pi}_K\triangleq
(\hat{\delta}_{1K},\hat{\delta}_{0K},\hat{\theta}_{0K})^\top$,
which is the lower right corner element of
$\lim_{N\rightarrow\infty}\mathbb{V}(\sqrt{N}(\hat{\tau}_K-\tau^*_K
))$.
Consider a sub-matrix formed by the last three rows of
$\mathbb{E}\left[\frac{\partial g_K(T,X,M,Y;\tau^*_K )}{\partial
    \tau}\right]^{-1}$, which is
\begin{equation*}
L_K \ \triangleq \ \left(C_{3\times 3K} A^{-1}_{3K\times 3K}\ , -I_{3
    \times 3} \right).
\end{equation*}
Algebraic manipulation yields:
\begin{equation*}
L_K \ = \
\begin{pmatrix}
L_{11K}&0_{1 \times K}&0_{1 \times K}& \\
0_{1 \times K}&L_{22K}&0_{1 \times K}&-I_{3 \times 3}\\
0_{1 \times K}&L_{32K}&L_{33K}
\end{pmatrix}\ ,
\end{equation*}
where
\begin{eqnarray*}
L_{11K} & =& \mathbb{E}[T\rho''((\phi_K^*)^\top u_K(X))Y
             u_K(X)^\top]\cdot\mathbb{E}[T\rho''((\phi_K^*)^\top u_K(X))u_K(X)u_K(X)^\top]^{-1} \ ,\\
L_{22K} & =&\mathbb{E}[(1-T)\rho''((\lambda_K^*)^\top u_K(X))Y
             u_K(X)^\top]\cdot
             \mathbb{E}[(1-T)\rho''((\lambda_K^*)^\top u_K(X))u_K(X)u_K(X)^\top]^{-1} \ ,\\
L_{32K} & =& \mathbb{E}[T\rho''((\beta_K^*)^\top v_K(X,M))Y
             v_K(X,M)^\top]\\
& & \cdot \mathbb{E}\left[T\rho''((\beta_K^*)^\top v_K(X,M))v_K(X,M)v_K(X,M)^\top\right]^{-1} \\
 & & \cdot\mathbb{E}[(1-T)\rho''((\lambda_K^*)^\top
     u_K(X))v_K(X,M)u_K(X)^\top]\\
& & \cdot \mathbb{E}[(1-T)\rho''((\lambda_K^*)^\top u_K(X))u_K(X)u_K(X)^\top]^{-1}\ ,\\
L_{33K}  & =& \mathbb{E}[T\rho''((\beta_K^*)^\top v_K(X,M))Y
              v_K(X,M)^\top]\\
        & & \cdot \mathbb{E}[T\rho''((\beta_K^*)^\top v_K(X,M))v_K(X,M)v_K(X,M)^\top]^{-1} \ .
\end{eqnarray*}
Applying Lemmas~\ref{lemma2}~and~\ref{lemma3}, we show in the
supplementary material that
 \begin{align} \label{eq:rate}
&\left\|\mathbb{E}[g_K(T,X,M,Y;\tau^*_K)]\right\|^2=O\left(K^{-\frac{s}{r_1}+1}\right)+O\left(K^{-\frac{s'}{r}+1}\right)\ ,\\
& \label{eq:rate1} \left\|L_K\right\|^2=O(K^4)\ .
\end{align}
By \eqref{eq:empirical}, \eqref{eq:op1}  we can  have
\begin{align} \label{eq:tau}
\hat{\tau}_K-\tau_K^*=\mathbb{E}\left[\frac{\partial g_K(T,X,M,Y;\tau_K^*)}{\partial \tau}\right]^{-1}\left(\frac{1}{N}\sum_{i=1}^Ng_K(T_i,X_i,M_i,Y_i;\hat{\tau}_K)\right)\ .
\end{align}
Also, by equations~\eqref{eq:rate}~and~\eqref{eq:rate1}, we can show:
\begin{align} \label{eq:ar1}
\lim_{K\rightarrow \infty}Var(\sqrt{N}(\hat{\pi}_K-\pi^*_K))
&=\lim_{K \rightarrow \infty}L_K\mathbb{E}[g_K(T,X,M,Y;\tau^*_K)g_K(T,X,M,Y;\tau^*_K)^\top]
L_K^\top \\ \notag
&=\lim_{K \rightarrow \infty}L_KP_KL_K^\top\ ,
\end{align}
where
$\pi_K^*\triangleq
\left(\delta^*_{1K},\delta^*_{0K},\theta^*_K\right)^\top$
and
$P_K \triangleq
\mathbb{E}[g_K(T,X,M,Y;\tau^*_K)g_K(T,X,M,Y;\tau^*_K)^\top]$.
This leads to the following estimator of the asymptotic variance:
\begin{eqnarray}
\widehat{V}_K & \triangleq & \widehat{L}_K\widehat{P}_K\widehat{L}_K^\top,
\end{eqnarray}
where 
\begin{align*}
\hat{L}_k& \ = \
\begin{pmatrix}
\widehat{L}_{11K}&0_{1 \times K}&0_{1 \times K}& \\
0_{1 \times K}&\widehat{L}_{22K}&0_{1 \times K}&-I_{3 \times 3}\\
0_{1 \times K}&\widehat{L}_{32K}&\widehat{L}_{33K}
\end{pmatrix}\ ,\\
\widehat{L}_{11K} & \ \triangleq \
                \left[\frac{1}{N}\sum_{i=1}^NT_i\rho''(\hat{\phi}_K^\top
                u_K(X_i))u_K(X_i)^\top Y_i\right] \cdot
                \left[\frac{1}{N}\sum_{i=1}^N
                T_i\rho''(\hat{\phi}_K^\top u_K(X_i))u_K(X_i)^\top u_K(X_i)\right]^{-1}\ ,\\
\widehat{L}_{22K} & \ \triangleq \
                \left[\frac{1}{N}\sum_{i=1}^N(1-T_i)\rho''(\hat{\lambda}_K^\top
                u_K(X_i))u_K(X_i)^\top Y_i\right]\\
                 & \quad \  \cdot\left[\frac{1}{N}\sum_{i=1}^N(1-T_i)\rho''(\hat{\lambda}_K^\top
                    u_K(X_i))u_K(X_i)^\top u_K(X_i)\right]^{-1}\ ,\\
\widehat{L}_{32K} & \ \triangleq \ \left[\frac{1}{N}\sum_{i=1}^N
                    T_i\rho''(\hat{\beta}_K^\top
                    v_K(X_i,M_i))v_K(X_i,M_i)^\top Y_i\right]\\
         &\quad \ \cdot\left[\frac{1}{N}\sum_{i=1}^NT_i\rho''(\hat{\beta}_K^\top
           v_K(X_i,M_i))v_K(X_i,M_i)v_K(X_i,M_i)^\top\right]^{-1}\\
         &\quad \  \cdot\left[\frac{1}{N}\sum_{i=1}^N(1-T_i)\rho''(\hat{\lambda}_K^\top
           u_K(X_i))v_K(X_i,M_i)u_K(X_i)^\top\right]\\
         &\quad \ \cdot\left[\frac{1}{N}\sum_{i=1}^N(1-T_i)\rho''(\hat{\lambda}_K^\top
           u_K(X_i))u_K(X_i)u_K(X_i)^\top\right]^{-1}\ ,\\
  \widehat{L}_{33K} & \ \triangleq \
                      \left[\frac{1}{N}\sum_{i=1}^N T_i\rho''(\hat{\beta}_K^\top
                      v_K(X_i,M_i))v_K(X_i,M_i)^\top
                      Y_i\right]\\
         &\quad \ \cdot
           \left[\frac{1}{N}\sum_{i=1}^NT_i\rho''(\hat{\beta}^{T}_Kv_K(X_i,M_i))v_K(X_i,M_i)^\top
           v_K(X_i,M_i)\right]^{-1}\ ,\\
  \widehat{P}_K & \ \triangleq \  \frac{1}{N}\sum_{i=1}^N g_K(T_i,X_i,M_i,Y_i;\hat{\tau}_K)g_K(T_i,X_i,M_i,Y_i;\hat{\tau}_K)^\top\ .
\end{align*}

Finally, the following theorem {\color{black}states} that this estimator of the
asymptotic variance $\widehat{V}_K$ is consistent, {\color{black}indeed}.
\begin{theorem}
  Let $k_1=(0, 0, 1), k_2=(1,0,-1)$ and $k_3=(0,-1,1)$. Under
  Assumptions~\ref{assump:SI}--\ref{as:rho},
  $k_1\widehat{V}_Kk_1^\top$ is a consistent estimator for
  $\mathbb{E}\left[(S_{\theta_0})^2\right]$,
  $k_2\widehat{V}_Kk_2^\top$ is a consistent estimator for
  $\mathbb{E}\left[(S_{\textsf{NIE}})^2\right]$ and
  $k_3\widehat{V}_Kk_3^\top$ is a consistent estimator for
  $\mathbb{E}\left[(S_{\textsf{NDE}})^2\right]$, where
  $ S_{\textsf{NIE}}= S_{\delta_1}-S_{\theta_0}$,
  $S_{\textsf{NDE}}= S_{\theta_0}-S_{\delta_0}$ and
  $S_{\theta_0}, S_{\delta_1}, S_{\delta_0}$ are defined in
  equations~\eqref{E:sefftheta0}~and~\eqref{E:seffdelta},
  respectively.
\end{theorem}
\begin{proof}
The proof is given in the supplementary material.
\end{proof}

\section{Extensions of our Proposed Methodology}
\label{sec:extensions}

In this section, we consider several extensions of the proposed
methodology presented in the previous section.  We first study the
case of multiple mediators, and then discuss the estimation of pure
indirect effects and natural indirect effects for the untreated.

\subsection{Multiple mediators}
\label{subsec:multiple}

We show that the proposed methodology introduced in the previous
section does not require the mediator to be univariate.  Instead, we
consider the situation where multiple mediators exist and are not
causally independent.  Specifically, we study the setting considered
by \cite{imai:yama:13} who proposed a semiparametric estimation method
within a linear structural equation modeling framework.  We improve
this existing method by considering the efficient nonparametric
estimation.

Consider the setting with a binary treatment variable $T\in\{0,1\}$
and two mediators $W$ and $M$ with $W$ being causally prior to $M$.
Let $W(t)$ {\color{black}be} a potential mediator variable, {\color{black}which represents} the
value of the mediator $W$ when the treatment variable is set to
$t\in\{0,1\}$.  Similarly, denote the potential mediator variable of
$M$ by $M(t,w)$ which represents the value of the mediator $M$ when
$T$ and $W$ are set to $(t,w)$.  Finally, the potential outcome
variable can be defined as $Y(t, m, w)$, which represents the value of
the outcome variable $Y$ when the treatment and the two mediators are
set to $(t, m, w)$.

We consider an extension of the sequential ignorability assumption
given in Assumption~\ref{assump:SI}.
\begin{assumption} \label{assump:SI2} The following three conditional
  independence statements hold:
\begin{enumerate}
\item $\left\{Y(t,m,w),M(t,w),W(t)\right\} \ \indep \ T \mid X=x$ \ ,
\item $\left\{Y(t^\prime,m,w),M(t^\prime,w)\right\} \ \indep \ W(t)\mid T=t,X=x$ \ ,
\item {\color{black}$Y(t^\prime,m,w) \ \indep \ M(t,w')\mid T=t,X=x,W(t)=w'$};
\end{enumerate}
{\color{black}for any $t,t'\in\{0,1\}$ and $(m,w,w',x)\in\mathcal{M}\times \mathcal{W}\times\mathcal{W}\times\mathcal{X}$, where $\mathcal{W}$ is the support of $W$}.
\end{assumption}
This assumption is stronger than the one considered by
\cite{imai:yama:13}, which avoids the ``cross-world'' independence but
results in partial identification.  Finally, we also make the relevant
consistency and positivity assumptions as done in
Assumption~\ref{assump:SI}.  The following lemma establishes an
important connection between two sequential ignorability assumptions
given in Assumptions~\ref{assump:SI}~and~\ref{assump:SI2}.

\begin{lemma}\label{lemma_e}
  Under the first two conditions of Assumption~\ref{assump:SI2}, the
  third condition of Assumption~\ref{assump:SI2}  holds if and only if:
\begin{align}\label{multi:seq3_equivalent}
  Y(t^\prime,w,m)\ \indep \ \{W(t),M(t,W(t))\} \mid T=t,X=x 
\end{align}
for any {\color{black}$t,t^\prime\in\{0,1\}$ and $(m,w,x)\in \mathcal{M}\times \mathcal{W}\times \mathcal{X}$}.
\end{lemma}
\begin{proof}
Under the first two conditions of Assumption~\ref{assump:SI2}, we first show that the equation~\eqref{multi:seq3_equivalent} implies the third condition of Assumption~\ref{assump:SI2}. For any test functions $\phi_1(y)$, $\phi_2(m)$, and $\phi_3(w)$, by using \eqref{multi:seq3_equivalent} and Assumption 9 (2), we have
\begin{align*}
&\mathbb{E}\left[\phi_1(Y(t',m,w))\phi_2(M(t,W(t)))\phi_3(W(t))|T=t,X=x\right]\\
=&\mathbb{E}\left[\phi_1(Y(t',m,w))|T=t,X=x\right]\cdot \mathbb{E}\left[\phi_2(M(t,W(t)))\phi_3(W(t))|T=t,X=x\right]\\
=& \mathbb{E}\bigg[\mathbb{E}\left[\phi_1(Y(t',m,w))|T=t,X=x\right]\\
&\quad \quad \cdot \mathbb{E}\left[\phi_2(M(t,W(t)))\phi_3(W(t))|T=t,X=x,W(t)\right]\bigg|T=t,X=x\bigg]\\
=& \mathbb{E}\bigg[\mathbb{E}\left[\phi_1(Y(t',m,w))|T=t,X=x,W(t)\right] \\
& \quad \cdot \mathbb{E}\left[\phi_2(M(t,W(t)))|T=t,X=x,W(t)\right]\cdot \phi_3(W(t))\bigg|T=t,X=x\bigg]\ ,
\end{align*} 
{\color{black}where the second equality follows from applying the Tower property to the second term and then plugging  the first term inside the second one; the last equality follows from Assumption 9 (2).  Therefore, we can get that}
$$Y(t^\prime,m,w) \ \indep \ M(t,W(t))\mid T=t,X=x,W(t)\ .$$

Conversely, we shall show that under the first two conditions of Assumption~\ref{assump:SI2}, Assumption 9 (3) implies \eqref{multi:seq3_equivalent}. For any test functions $\phi_1(y)$ and $\phi_4(w,m)$, we have
\begin{align*}
&\mathbb{E}\left[\phi_1(Y(t',w,m))\phi_4(W(t),M(t,W(t)))|T=t,X=x\right]\\
=&\mathbb{E}\left[\mathbb{E}\left[\phi_1(Y(t',w,m))\phi_4(W(t),M(t,W(t)))|T=t,X=x,W(t)\right]|T=t,X=x\right]\\
=&\mathbb{E}\left[\mathbb{E}\left[\phi_1(Y(t',w,m))|T=t,X=x,W(t)\right]\right.\\ 
&\quad \cdot \left.\mathbb{E}\left[\phi_4(W(t),M(t,W(t)))|T=t,X=x,W(t)\right]|T=t,X=x\right] \\
=&\mathbb{E}\left[\phi_1(Y(t',w,m))|T=t,X=x\right] \\
& \quad \cdot  \mathbb{E}\left[\mathbb{E}\left[\phi_4(W(t),M(t,W(t)))|T=t,X=x,W(t)\right]|T=t,X=x\right]\\
=&\mathbb{E}\left[\phi_1(Y(t',w,m))|T=t,X=x\right] \cdot \mathbb{E}\left[\phi_4(W(t),M(t,W(t)))|T=t,X=x\right]\ ,
\end{align*}
where the second equality follows from Assumption 9 (3) and the third equality follows from Assumption 9 (2), which finally yields \eqref{multi:seq3_equivalent}.
\end{proof}

Lemma~\ref{lemma_e}, together with the discussion in the previous
section, implies that we can estimate $\E[Y(1,W(0),M(0,W(0)))]$ using
the proposed methodology that calibrates the functions of $X$, $W$ and
$M$.  To study path-specific effects, we consider the following
decomposition of the average treatment effects:{\color{black}
\begin{align}
&\E[Y(1,W(1),M(1,W(1)))-Y(0,W(0),M(0,W(0)))]\notag\\
= \ &\E[Y(1,W(1),M(1,W(1)))-Y(1,W(0),M(1,W(0)))]\label{eq:partial.mediation1}\\
&+\E[Y(1,W(0),M(1,W(0)))-Y(1,W(0),M(0,W(0)))]\label{eq:partial.mediation2}\\
&+\E[Y(1,W(0),M(0,W(0)))-Y(0,W(0),M(0,W(0)))]\ .
\end{align}}
This decomposition is also studied by \cite{avin:etal:05} and
\cite{vanderweele2014effect}.  The first term represents a partial
mediation effect through $W$, the second term is the partial mediation
effect through $M$, and the third term is the direct effect of the
treatment that does not go through $W$ or $M$.  Although the above
decomposition is not the only way to define partial mediation effects,
it has several advantages.  First, the sum of the two partial
mediation effects equals the joint natural indirect effects through
both mediators.  Moreover, the first term and the sum of the second
and third terms are identified even if $M$ is not observed.

{\color{black}To estimate the two partial mediation effects \eqref{eq:partial.mediation1} and \eqref{eq:partial.mediation2}}, we show below that
$\E[Y(1,W(0),M(1,W(0)))]$ is estimable by adapting the proposed
methodology under Assumption~\ref{assump:SI2}.  The following lemma
presents the results in {\color{black}the} current setting that are analogous to
those given in Lemma~\ref{lemma1} in the univariate mediator case.
\begin{lemma}\label{lemma_f}
  Let $q_0(X)=(Nf_{T\mid X}(0\mid X))^{-1}$ and
  $r_0(X,W)= f_{T\mid X,W}(0\mid X,W) \cdot(N f_{T\mid X,W}(1\mid
  X,W)f_{T\mid X}(0\mid X))^{-1}$.
  Under Assumption~\ref{assump:SI2}, for any suitable function
  $v(X,M)$, we have the following properties:
\begin{align}
  &\E[Y(1,W(0),M(1,W(0)))] \ = \ \E[TNr_0(X,W)Y] \ ,  \label{id:W} \\
  &\E[TNr_0(X,W)v(X,W)] \ = \ \E[(1-T)Nq_0(X)v(X,W)] \ .\label{moment:W} 
\end{align}
\end{lemma}
\begin{proof}
  We begin by noting that
\begin{align}
&\mathbb{E}\left[Y(1,W(0),M(1,W(0)))\mid X=x\right] \notag\\
=&\int_{\mathcal{W}\times \mathcal{M}}\mathbb{E}\left[Y(1,w,m)\mid
   X=x,W(0)=w,M(1,W(0))=m\right]\notag\\
& \quad \cdot f_{W(0),M(1,W(0))\mid X}(w,m\mid x)dwdm \label{condition_theta_0}
\end{align}
We can express each term in this equation using the observed data, {\color{black}
\begin{align} 
&\mathbb{E}\left[Y(1,w,m)\mid X=x,W(0)=w,M(1,W(0))=m\right] \label{conditional_Y} \\
=&\mathbb{E}\left[Y(1,w,m)\mid X=x,W(0)=w,M(1,w)=m\right]\notag\\
=&\mathbb{E}\left[Y(1,w,m)\mid X=x,W(0)=w,M(1,w)=m,T=0\right] \quad   \text{(by Assumption 9 (1))} \notag\\
=& \mathbb{E}\left[Y(1,w,m)\mid X=x,M(1,w)=m,T=0\right] \quad \text{(by Assumption 9 (2))} \notag\\
=&\mathbb{E}\left[Y(1,w,m)\mid X=x,M(1,w)=m,T=1\right] \quad \text{(by Assumption 9 (1))} \notag\\
=&\mathbb{E}\left[Y(1,w,m)\mid X=x,W(1)=w,M(1,w)=m,T=1\right]  \quad   \text{(by Assumption 9 (2))} \notag\\
=&\mathbb{E}\left[Y(1,w,m)\mid X=x,W(1)=w,M(1,W(1))=m,T=1\right]\notag\\
=&\mathbb{E}\left[Y\mid X=x,W=w,M=m,T=1\right]\ . \notag
\end{align}}
Similarly, we have {\color{black}
\begin{align}
&f_{W(0),M(1,W(0))\mid X}(w,m\mid x)\label{Delta}\\
=& f_{M(1,W(0))\mid X,W(0)}(m\mid x,w)f_{W(0)\mid X}(w\mid x) \notag\\
=&f_{M(1,w)\mid X,W(0)}(m\mid x,w)f_{W(0)\mid X}(w\mid x) \notag \\
=&f_{M(1,w)\mid X,W(0),T}(m\mid x,w,0)f_{W(0)\mid X}(w\mid x) \quad   \text{(by Assumption 9 (1))}\notag\\
=&f_{M(1,w)\mid X,W(0),T}(m\mid x,w,0)f_{W(0)\mid X,T}(w\mid x,0) \quad   \text{(by Assumption 9 (1))} \notag \\
=& f_{M(1,w)\mid X,T}(m\mid x,0)f_{W(0)\mid X,T}(w\mid x,0) \quad   \text{(by Assumption 9 (2))}\notag\\
=&f_{M(1,w)\mid X,T}(m\mid x,1)f_{W(0)\mid X,T}(w\mid x,0)\quad   \text{(by Assumption 9 (1))}\notag\\
=&f_{M(1,w)\mid X,W(1),T}(m\mid x,w,1)f_{W(0)\mid X,T}(w\mid x,0) \quad   \text{(by Assumption 9 (2))}\notag\\
=&f_{M(1,W(1))\mid X,W(1),T}(m\mid x,w,1)f_{W(0)\mid X,T}(w\mid x,0)\notag\\
=&f_{M\mid X,W,T}(m\mid x,w,1)f_{W\mid X,T}(w\mid x,0) \ . \notag
\end{align}}
Then, equations~\eqref{condition_theta_0},~\eqref{conditional_Y}~and~\eqref{Delta} imply that
\begin{align*}
&\mathbb{E}\left[Y(1,W(0),M(1,W(0)))\mid X=x\right]\\
=&\int_{\mathcal{W}\times \mathcal{M}} \mathbb{E}(Y\mid
   X=x,W=w,M=m,T=1) f_{M\mid X,W,T}(m\mid x,w,1)f_{W\mid X,T}(w\mid x,0)dwdm \\
=&\sum_{t=0}^1\int_{\mathcal{Y}\times\mathcal{W}\times \mathcal{M}}ty \cdot f_{Y\mid X,W,M,T}(y\mid x,w,m,t)f_{M\mid X,W,T}(m\mid x,w,1)f_{W\mid X,T}(w\mid x,0)dydwdm\\
=&\sum_{t=0}^1\int_{\mathcal{Y}\times\mathcal{W}\times \mathcal{M}}ty
   \cdot \frac{f_{Y,X,W,M,T}(y,x,w,m,t)}{f_{X,W,M,T}(x,w,m,1)}\cdot
   \frac{f_{M,X,W,T}(m,x,w,1)}{f_{X,W,T}(x,w,1)} \cdot \frac{f_{W,X,T}(w,x,0)}{f_{X,T}(x,0)}dydwdm\\
=&\sum_{t=0}^1\int_{\mathcal{Y}\times\mathcal{W}\times \mathcal{M}}ty \cdot \frac{f_{Y,X,W,M,T}(y,x,w,m,t)}{f_{X,W,T}(x,w,1)}\cdot \frac{f_{W,X,T}(w,x,0)}{f_{X,T}(x,0)}dydwdm \ .
\end{align*}
Therefore, we can have
\begin{align}
&\mathbb{E}\left[Y(1,W(0),M(1,W(0)))\right]\notag\\
= \ &\sum_{t=0}^1\int_{\mathcal{X}\times\mathcal{W}\times \mathcal{M}\times\mathcal{Y}} ty \cdot \frac{f_{Y,X,W,M,T}(y,x,w,m,t)}{f_{X,W,T}(x,w,1)}\cdot \frac{f_{W,X,T}(w,x,0)}{f_{X,T}(x,0)}f(x)dxdwdmdy\notag\\
= \ &\mathbb{E}\left[TY\frac{f_{X,W,T}(X,W,0)}{f_{X,W,T}(X,W,1)f_{T\mid X}(0\mid X)}\right]\notag\\
=\ &\mathbb{E}\left[TYNr_0(X,W)\right]\notag\ .
\end{align}
This  proves equation~\eqref{id:W}.  For equation~\eqref{moment:W}, we
can show that the term on the left-hand side is equal to,
\begin{align}
 &\mathbb{E}[TNr_0(X,W)v(X,W)]\notag\\
= \ &\int_{\mathcal{X}\times \mathcal{W}}\frac{f_{T\mid X,W}(0\mid x,w)}{f_{T\mid X,W}(1\mid x,w)f_{T\mid X}(0\mid x)}v(x,w)f_{T,X,W}(1,x,w)dxdw\notag\\
= \ &\int_{\mathcal{X}\times \mathcal{W}}\frac{f_{T\mid X,W}(0\mid x,w)}{f_{T\mid X}(0\mid x)}v(x,w)f_{X,W}(x,w)dxdw\notag\\
= \ &\int_{\mathcal{X}\times \mathcal{W}}\frac{1}{f_{T\mid X}(0\mid x)}v(x,w)f_{T,X,W}(0,x,w)dxdw\notag\\
= \ &\sum_{t=0}^1\int_{\mathcal{X}\times \mathcal{W}}\frac{1-t}{f_{T\mid X}(0\mid x)}v(x,w)f_{T,X,W}(t,x,w)dxdw\notag\\
= \ &\mathbb{E}\left[(1-T)Nq_0(X)v(X,W)\right] \ .\notag
\end{align}
This  proves equation~\eqref{moment:W}.
\end{proof}

In Lemma~\ref{lemma_f}, equation~\eqref{id:W} expresses the average
potential outcome, $\E[Y(1,W(0),M(1,W(0)))]$, in terms of observable
variables; and equation~\eqref{moment:W} gives an identification
condition of the unknown weight function $r_0(x,w)$.  Therefore,
$\E[Y(1,W(0),M(1,W(0)))]$ can be estimated by the proposed methodology
described in the previous section that calibrates the functions of $X$
and $W$ but not $M$.

\subsection{Natural direct effect for the untreated}

\cite{lendle2013identification} {\color{black}considered} the estimation of the natural
direct effect for the untreated (NDEU), which is defined as
$\textsf{NDEU}=\E[Y(1,M(0))-Y(0,M(0))\mid T=0]=\theta'_1-\delta'_0$.
\cite{lendle2013identification} {\color{black}showed} that the identification of NDEU
requires the following assumption (along with {\color{black}Assumption \ref{assump:SI}}).
\begin{assumption}\label{as:NPEU} For any $t \in \left\{0,1\right\}$,
  $m,m'\in\mathcal{M}$ and $x\in\mathcal{X}$, the following equality
  holds,
$$\mathbb{E}\left[Y(t,m)-Y(0,m)\mid M(0)=m',X=x\right]\ = \ \mathbb{E}\left[Y(t,m)-Y(0,m)\mid X=x\right].$$
\end{assumption}
Under Assumption \ref{as:NPEU}, we have
\begin{align*}
\textsf{NDEU} &\ \triangleq\ \mathbb{E}\left[Y(1,M(0))-Y(0,M(0))\mid T=0\right]\\
& \ = \ \mathbb{E}\left[\frac{TYf_{M,X\mid T}(M,X\mid 0)}{f_{T,M,X}(1,M,X)}\right]-\frac{\mathbb{E}\left[(1-T)Y\right]}{\mathbb{E}(1-T)} \ .
\end{align*}
Define
$$\tilde{r}_0(x,m) \ = \ \frac{f_{M,X\mid T}(m,x\mid
  0)}{Nf_{T,M,X}(1,m,x)}\ ,$$
for any integrable function $v(x,m)$, we have,
\begin{align*}
\mathbb{E}\left[T v(X,M)N\tilde{r}_0(X,M)\right] \ =\ & \int v(x,m)\frac{f_{M,X\mid T}(m,x\mid 0)}{f_{T,M,X}(1,m,x)}\cdot f_{T,M,X}(1,m,x)dmdx\\
\ =\ &\int v(x,m)f_{M,X\mid T}(m,x\mid 0)dmdx\\
\ =\ & \frac{\mathbb{E}[(1-T)v(X,M)]}{\mathbb{E}(1-T)}.
\end{align*}
Therefore, we define
$$\hat{\tilde{r}}_K(x,m)\ = \ \frac{1}{N}\rho^\prime \left(\hat{{\beta}}_{1K}^\top v_K(x,m)\right)\
,$$
where $\hat{\beta}_{1K}$ {\color{black}maximizes the} following general empirical
likelihood function $\widehat{H}_{1K}(\beta)$,
\begin{align*}
  \widehat{H}_{1K}(\beta) \ \triangleq \ \frac{1}{N}\sum_{i=1}^N
  \left\{T_i\rho\left(\beta^\top v_K(X_i,M_i)\right)- \frac{(1-T_i)\beta^\top v_K(X_i,M_i)}{1-\overline{T}}\right\}\ ,
\end{align*}
where $\overline{T}\triangleq \frac{1}{N}\sum_{i=1}^N T_i$. Therefore, we can define the estimators of $\theta_1'$ and $\delta_0'$ to be,
\begin{align*}
  \hat{\theta}_{1K}' & \ \triangleq \ \sum_{i=1}^N T_i\hat{\tilde{r}}_K(X_i,M_i)Y_i \ ,\\
  \hat{\delta}_0' & \ \triangleq \ \frac{\sum_{i=1}^N(1-T_i)Y_i}{\sum_{i=1}^N(1-T_i)}  \ .
 \end{align*}
Given these estimators, the proposed estimator for the NDEU is given by,
\begin{align*}
  \widehat{\textsf{NDEU}} & \ \triangleq \ \hat{\theta}_{1K}'- \hat{\delta}_0'
\end{align*}
The proposed estimator $\widehat{\textsf{NDEU}}$ extends the
estimators of average treatment effects on the treated studied by
\cite{hain:12} and \cite{chan2015lobally} to causal mediation
analysis.  The following theorem summarizes the asymptotic properties
of this estimator.
\begin{theorem}
 {\color{black} Under Assumptions 1-8 and 10}, $\widehat{\textsf{NDEU}}$ has
  the following properties:
\begin{enumerate}
\item $\widehat{\textsf{NDEU}} \stackrel{p}{\longrightarrow} \textsf{NDEU}$
\item
  $ \dps \sqrt{N}\left(\widehat{\textsf{NDEU}}- \textsf{NDEU}\right)
  \stackrel {d} {\longrightarrow} \mathcal{N}(0,
  V_{\textsf{\textsf{NDEU}}})$,
  where
  $V_{\textsf{\textsf{NDEU}}} =
  \mathbb{E}\left(S_{\textsf{\textsf{NDEU}}}^2\right)$
  attained the semi-parametric efficiency bound
  \citep{lendle2013identification}, where 
\begin{align*}
  S_{\textsf{NDEU}}\ = \ &\left\{\frac{\mathbf{1}\{T=1\}}{f_{T}(0)}\frac{f_{T\mid X,M}(0\mid X,M)}{f_{T\mid X,M}(1\mid X,M)}-\frac{\mathbf{1}\{T=0\}}{f_T(0)}\right\}\left\{Y-\mathbb{E}\left(Y\mid T,M,X\right)\right\} \\
                    &+\frac{\mathbf{1}\{T=0\}}{f_{T}(0)}\left\{\mathbb{E}\left(Y\mid T=1,M,X\right)-\mathbb{E}\left(Y\mid T=0,M,X\right)-\textsf{NDEU}\right\} \ .
\end{align*}
\end{enumerate}
\end{theorem}

\section{Empirical Applications}
\label{sec:dataanalysis}

In this section, we briefly describe the application of the proposed
methodology to two data sets, one regarding the evaluation of a job
training program, {\color{black}while another} concerning a psychological experiment
from political science.

\subsection{Evaluation of a job training program}

We reanalyze the data from the JOBSII intervention study for
unemployed job seekers \citep{vinokur1997mastery}.  JOBSII is a
randomized intervention study that {\color{black}investigates} the efficacy of a job
training program.  In the study, 1801 unemployed workers were randomly
assigned into two groups, with 1249 workers in the intervention arm
and 552 workers in the control arm.  The intervention group received
five job training workshops that had a specific focus on improving a
general sense of mastery of job seekers, which is a composite measure
including confidence, internal locus of control, and {\color{black}self-esteem}.  It
was hypothesized that the sense of mastery was a key mediator of a
successful reemployment.  The control group received a booklet but not
the workshop sessions.  \citet{vinokur1997mastery} used a linear
structural equation model and found that the enhanced sense of mastery
had a significant mediating effect on reemployment.  

To relax the linearity assumption, we reanalyze the JOBSII data using
the {\color{black}present} proposed nonparametric method.  We include the following baseline
covariates in the analysis: financial strain, depressive symptoms,
age, sex, race, education, marital status, previous income, and
previous occupation.  The function $u(X)$ is set to be linear in
covariates, and $v(X,M)$ is linear in covariates and the mediator.  We
find that the overall average treatment effect is 6.4\% with a 95\%
confidence interval of [0.9\%, 11.8\%].  The estimated average natural
indirect effect is 0.8\% {\color{black}(resp. [0.0\%, 1.6\%])}, while the estimated average
direct effect is {\color{black}5.6\% (resp. [0.1\%,
11.1\%])}.  Our result suggests that the improvement in the sense of
mastery contributed approximately 12.5\% to the total average
treatment effect.  This finding is stronger than the original study,
which reported the estimated proportion of this mediating effect to be
7.2\%.

\subsection{A political framing study}

Next, we reanalyze the data from an experiment conducted by
\cite{brad:vale:suha:08}, who studied the role of emotions for the
effect of news stories on the preferences of immigration policies.
The authors randomly assigned 351 White non-Latino adults into two
groups.  In one group, respondents were shown a {\color{black}newspaper} article
with a picture of a Hispanic immigrant.  In the other group, the same
article was shown except with a picture of an European immigrant.  The
original authors studied the mediation of the framing effect through
two psychological mechanisms: the perceived harm mechanism and the
anxiety mechanism.  They assumed that the two mediators are
independent but the assumption may not be plausible because an
increased level of perceived harm of immigration can cause more
anxiety about immigration policies.  Therefore, we apply the proposed
methodology described in Section~\ref{subsec:multiple}, which allows
for multiple mediators.  Specifically, we use the level of perceived
harm as $W$ while the level of anxiety is $M$. We include the same set
of covariates as in the original paper, which includes education, age,
income, and gender.

We find that the news article with a picture of a Latino immigrant
leads to a 0.430 percentage point increase (relative to the same
story with a picture of an European immigrant) in the opposition of
increased immigration on the five-point scale, with a 95\% confidence
interval of [0.199, 0.661].  The combined mediation effect from the
perceived harm and anxiety mechanisms is estimated to be 0.164
({\color{black}resp.} [0.001, 0.327]).  This implies that the two mediators together
account for 38\% of the total average treatment effect.  We further
decompose this combined mediation effect into two partial mediation
effects.  The partial indirect effect concerning the perceived harm
mechanism is estimated to be 0.115 ({\color{black}resp.} [$-$0.03, 0.263]) whereas the
estimated partial mediation effect for the anxiety mechanism is 0.049
({\color{black}resp.}  [$-$0.159, 0.257]).  This finding suggests that the anxiety mechanism
may play only a secondary role in the public opinion about immigration
policy.  Our results contradict with the original study, which
concluded that the anxiety plays an essential role by assuming {\color{black}that} the two mediators are causally independent. 

\section{Concluding remarks}
\label{sec:conclusion}

In this paper, we proposed a novel methodology for causal mediation
analysis.  We establish that the proposed estimator is {\color{black}fully} nonparametric
and globally semiparametric efficient.  This improves the existing
estimators, which rely upon parametric assumptions and is only locally
efficient.  Furthermore, we show how to consistently estimate the
asymptotic variance of this proposed estimator without requiring
additional nonparametric estimation.  Another advantage is the
availability of efficient and stable numerical algorithms that can be
used to compute the proposed estimator and its estimated variance.  We
show how to extend our methodology to a setting with multiple
mediators and the estimation of related causal quantities of interest.
While causal mediation analysis has gained popularity in a variety of
disciplines, applied researchers have largely relied upon parametric
methods.  We believe that our nonparametric estimator has a potential
to significantly improve the credibility of causal mediation analysis
in scientific research.

\bibliographystyle{imsart-nameyear}
\bibliography{NPmediation,my,imai}

\begin{thebibliography}{43}

\bibitem[\protect\citeauthoryear{Albert}{2008}]{albe:08}
\begin{barticle}[author]
\bauthor{\bsnm{Albert},~\bfnm{Jeffrey~M.}\binits{J.~M.}}
(\byear{2008}).
\btitle{Mediation analysis via potential outcomes models}.
\bjournal{Statistics in Medicine}
\bvolume{27}
\bpages{1282--1304}.
\end{barticle}
\endbibitem

\bibitem[\protect\citeauthoryear{Avin, Shpitser and Pearl}{2005}]{avin:etal:05}
\begin{binproceedings}[author]
\bauthor{\bsnm{Avin},~\bfnm{Chen}\binits{C.}},
  \bauthor{\bsnm{Shpitser},~\bfnm{Ilya}\binits{I.}} \AND
  \bauthor{\bsnm{Pearl},~\bfnm{Judea}\binits{J.}}
(\byear{2005}).
\btitle{Identifiability of Path-Specific Effects}.
In \bbooktitle{Proceedings of the Nineteenth International Joint Conference on
  Artificial Intelligence}
\bpages{357--363}.
\bpublisher{Morgan Kaufmann}, \baddress{Edinburgh, Scotland}.
\end{binproceedings}
\endbibitem

\bibitem[\protect\citeauthoryear{Baron and Kenny}{1986}]{baron1986moderator}
\begin{barticle}[author]
\bauthor{\bsnm{Baron},~\bfnm{Reuben~M}\binits{R.~M.}} \AND
  \bauthor{\bsnm{Kenny},~\bfnm{David~A}\binits{D.~A.}}
(\byear{1986}).
\btitle{The moderator--mediator variable distinction in social psychological
  research: Conceptual, strategic, and statistical considerations.}
\bjournal{Journal of personality and social psychology}
\bvolume{51}
\bpages{1173--1182}.
\end{barticle}
\endbibitem

\bibitem[\protect\citeauthoryear{Brader, Valentino and
  Suhay}{2008}]{brad:vale:suha:08}
\begin{barticle}[author]
\bauthor{\bsnm{Brader},~\bfnm{Ted}\binits{T.}},
  \bauthor{\bsnm{Valentino},~\bfnm{Nicholas}\binits{N.}} \AND
  \bauthor{\bsnm{Suhay},~\bfnm{Elizabeth}\binits{E.}}
(\byear{2008}).
\btitle{What Triggers Public Opposition to Immigration? Anxiety, Group Cues,
  and Immigration Threat}.
\bjournal{American Journal of Political Science}
\bvolume{52}
\bpages{959--978}.
\end{barticle}
\endbibitem

\bibitem[\protect\citeauthoryear{Chan and Yam}{2014}]{chan2014oracle}
\begin{barticle}[author]
\bauthor{\bsnm{Chan},~\bfnm{Kwun Chuen~Gary}\binits{K.~C.~G.}} \AND
  \bauthor{\bsnm{Yam},~\bfnm{Sheung Chi~Phillip}\binits{S.~C.~P.}}
(\byear{2014}).
\btitle{Oracle, Multiple Robust and Multipurpose Calibration in a Missing
  Response Problem}.
\bjournal{Statistical Science}
\bvolume{29}
\bpages{380--396}.
\end{barticle}
\endbibitem

\bibitem[\protect\citeauthoryear{Chan, Yam and Zhang}{2015}]{chan2015lobally}
\begin{barticle}[author]
\bauthor{\bsnm{Chan},~\bfnm{KCG}\binits{K.}},
  \bauthor{\bsnm{Yam},~\bfnm{SCP}\binits{S.}} \AND
  \bauthor{\bsnm{Zhang},~\bfnm{Z}\binits{Z.}}
(\byear{2015}).
\btitle{Globally Efficient Nonparametric Inference of Average Treatment Effects
  by Empirical Balancing Calibration Weighting}.
\bjournal{To appear in Journal of Royal Statistical Society: Series B}.
\end{barticle}
\endbibitem

\bibitem[\protect\citeauthoryear{Chen, Hong and
  Tarozzi}{2008}]{chen2008semiprametric}
\begin{barticle}[author]
\bauthor{\bsnm{Chen},~\bfnm{X}\binits{X.}},
  \bauthor{\bsnm{Hong},~\bfnm{H}\binits{H.}} \AND \bauthor{\bsnm{Tarozzi}}
(\byear{2008}).
\btitle{Semiparametric efficiency in GMM models with auxiliary data}.
\bjournal{The Annals of Statistics}
\bvolume{36}
\bpages{808-843}.
\end{barticle}
\endbibitem

\bibitem[\protect\citeauthoryear{Deville and
  S{\"a}rndal}{1992}]{deville1992calibration}
\begin{barticle}[author]
\bauthor{\bsnm{Deville},~\bfnm{Jean-Claude}\binits{J.-C.}} \AND
  \bauthor{\bsnm{S{\"a}rndal},~\bfnm{Carl-Erik}\binits{C.-E.}}
(\byear{1992}).
\btitle{Calibration estimators in survey sampling}.
\bjournal{Journal of the American statistical Association}
\bvolume{87}
\bpages{376--382}.
\end{barticle}
\endbibitem

\bibitem[\protect\citeauthoryear{Geneletti}{2007}]{gene:07}
\begin{barticle}[author]
\bauthor{\bsnm{Geneletti},~\bfnm{Sara}\binits{S.}}
(\byear{2007}).
\btitle{Identifying direct and indirect effects in a non-counterfactual
  framework}.
\bjournal{Journal of the Royal Statistical Society, Series {B} (Statistical
  Methodology)}
\bvolume{69}
\bpages{199--215}.
\end{barticle}
\endbibitem

\bibitem[\protect\citeauthoryear{Graham, Pinto and
  Egel}{2012}]{grah:pint:egel:12}
\begin{barticle}[author]
\bauthor{\bsnm{Graham},~\bfnm{Bryan~S.}\binits{B.~S.}},
  \bauthor{\bsnm{Pinto},~\bfnm{{Cristine Campos de Xavier}}\binits{C.}} \AND
  \bauthor{\bsnm{Egel},~\bfnm{Daniel}\binits{D.}}
(\byear{2012}).
\btitle{Inverse probability tilting for moment condition models with missing
  data}.
\bjournal{Review of Economic Studies}
\bvolume{79}
\bpages{1053--1079}.
\end{barticle}
\endbibitem

\bibitem[\protect\citeauthoryear{Hahn}{1998}]{hahn1998role}
\begin{barticle}[author]
\bauthor{\bsnm{Hahn},~\bfnm{Jinyong}\binits{J.}}
(\byear{1998}).
\btitle{On the role of the propensity score in efficient semiparametric
  estimation of average treatment effects}.
\bjournal{Econometrica}
\bpages{315--331}.
\end{barticle}
\endbibitem

\bibitem[\protect\citeauthoryear{Hainmueller}{2012}]{hain:12}
\begin{barticle}[author]
\bauthor{\bsnm{Hainmueller},~\bfnm{Jens}\binits{J.}}
(\byear{2012}).
\btitle{Entropy Balancing for Causal Effects: Multivariate Reweighting Method
  to Produce Balanced Samples in Observational Studies}.
\bjournal{Political Analysis}
\bvolume{20}
\bpages{25--46}.
\end{barticle}
\endbibitem

\bibitem[\protect\citeauthoryear{Han and Wang}{2013}]{han2013estimation}
\begin{barticle}[author]
\bauthor{\bsnm{Han},~\bfnm{Peisong}\binits{P.}} \AND
  \bauthor{\bsnm{Wang},~\bfnm{Lu}\binits{L.}}
(\byear{2013}).
\btitle{Estimation with missing data: beyond double robustness}.
\bjournal{Biometrika}
\bvolume{100}
\bpages{417--430}.
\end{barticle}
\endbibitem

\bibitem[\protect\citeauthoryear{Hansen, Heaton and
  Yaron}{1996}]{hansen1996finite}
\begin{barticle}[author]
\bauthor{\bsnm{Hansen},~\bfnm{Lars~Peter}\binits{L.~P.}},
  \bauthor{\bsnm{Heaton},~\bfnm{John}\binits{J.}} \AND
  \bauthor{\bsnm{Yaron},~\bfnm{Amir}\binits{A.}}
(\byear{1996}).
\btitle{Finite-sample properties of some alternative GMM estimators}.
\bjournal{Journal of Business \& Economic Statistics}
\bvolume{14}
\bpages{262--280}.
\end{barticle}
\endbibitem

\bibitem[\protect\citeauthoryear{Hirano, Imbens and
  Ridder}{2003}]{hirano2000efficient}
\begin{barticle}[author]
\bauthor{\bsnm{Hirano},~\bfnm{Keisuke}\binits{K.}},
  \bauthor{\bsnm{Imbens},~\bfnm{Guido}\binits{G.}} \AND
  \bauthor{\bsnm{Ridder},~\bfnm{Geert}\binits{G.}}
(\byear{2003}).
\btitle{Efficient estimation of average treatment effects using the estimated
  propensity score}.
\bjournal{Econometrica}
\bvolume{71}
\bpages{1161--1189}.
\end{barticle}
\endbibitem

\bibitem[\protect\citeauthoryear{Imai, Keele and
  Tingley}{2010}]{imai:keel:ting:10}
\begin{barticle}[author]
\bauthor{\bsnm{Imai},~\bfnm{Kosuke}\binits{K.}},
  \bauthor{\bsnm{Keele},~\bfnm{Luke}\binits{L.}} \AND
  \bauthor{\bsnm{Tingley},~\bfnm{Dustin}\binits{D.}}
(\byear{2010}).
\btitle{A General Approach to Causal Mediation Analysis}.
\bjournal{Psychological Methods}
\bvolume{15}
\bpages{309--334}.
\end{barticle}
\endbibitem

\bibitem[\protect\citeauthoryear{Imai, Keele and
  Yamamoto}{2010}]{imai:keel:yama:10}
\begin{barticle}[author]
\bauthor{\bsnm{Imai},~\bfnm{Kosuke}\binits{K.}},
  \bauthor{\bsnm{Keele},~\bfnm{Luke}\binits{L.}} \AND
  \bauthor{\bsnm{Yamamoto},~\bfnm{Teppei}\binits{T.}}
(\byear{2010}).
\btitle{Identification, Inference, and Sensitivity Analysis for Causal
  Mediation Effects}.
\bjournal{Statistical Science}
\bvolume{25}
\bpages{51--71}.
\end{barticle}
\endbibitem

\bibitem[\protect\citeauthoryear{Imai and Ratkovic}{2014}]{imai:ratk:14}
\begin{barticle}[author]
\bauthor{\bsnm{Imai},~\bfnm{Kosuke}\binits{K.}} \AND
  \bauthor{\bsnm{Ratkovic},~\bfnm{Marc}\binits{M.}}
(\byear{2014}).
\btitle{Covariate Balancing Propensity Score}.
\bjournal{Journal of the Royal Statistical Society, {Series B} (Statistical
  Methodology)}
\bvolume{76}
\bpages{243--263}.
\end{barticle}
\endbibitem

\bibitem[\protect\citeauthoryear{Imai and Yamamoto}{2013}]{imai:yama:13}
\begin{barticle}[author]
\bauthor{\bsnm{Imai},~\bfnm{Kosuke}\binits{K.}} \AND
  \bauthor{\bsnm{Yamamoto},~\bfnm{Teppei}\binits{T.}}
(\byear{2013}).
\btitle{Identification and Sensitivity Analysis for Multiple Causal Mechanisms:
  Revisiting Evidence from Framing Experiments}.
\bjournal{Political Analysis}
\bvolume{21}
\bpages{141--171}.
\end{barticle}
\endbibitem

\bibitem[\protect\citeauthoryear{Imai et~al.}{2011}]{imai:etal:11}
\begin{barticle}[author]
\bauthor{\bsnm{Imai},~\bfnm{Kosuke}\binits{K.}},
  \bauthor{\bsnm{Keele},~\bfnm{Luke}\binits{L.}},
  \bauthor{\bsnm{Tingley},~\bfnm{Dustin}\binits{D.}} \AND
  \bauthor{\bsnm{Yamamoto},~\bfnm{Teppei}\binits{T.}}
(\byear{2011}).
\btitle{Unpacking the Black Box of Causality: Learning about Causal Mechanisms
  from Experimental and Observational Studies}.
\bjournal{American Political Science Review}
\bvolume{105}
\bpages{765--789}.
\end{barticle}
\endbibitem

\bibitem[\protect\citeauthoryear{Imbens, Newey and
  Ridder}{2005}]{imbens2005mean}
\begin{barticle}[author]
\bauthor{\bsnm{Imbens},~\bfnm{Guido~W}\binits{G.~W.}},
  \bauthor{\bsnm{Newey},~\bfnm{Whitney~K}\binits{W.~K.}} \AND
  \bauthor{\bsnm{Ridder},~\bfnm{Geert}\binits{G.}}
(\byear{2005}).
\btitle{Mean-square-error calculations for average treatment effects}.
\bjournal{IEPR working paper}.
\end{barticle}
\endbibitem

\bibitem[\protect\citeauthoryear{Jo}{2008}]{jo:08}
\begin{barticle}[author]
\bauthor{\bsnm{Jo},~\bfnm{Booil}\binits{B.}}
(\byear{2008}).
\btitle{Causal Inference in Randomized Experiments with Mediational Processes}.
\bjournal{Psychological Methods}
\bvolume{13}
\bpages{314--336}.
\end{barticle}
\endbibitem

\bibitem[\protect\citeauthoryear{Joffe et~al.}{2007}]{joffe2007defining}
\begin{barticle}[author]
\bauthor{\bsnm{Joffe},~\bfnm{Marshall~M}\binits{M.~M.}},
  \bauthor{\bsnm{Small},~\bfnm{Dylan}\binits{D.}},
  \bauthor{\bsnm{Hsu},~\bfnm{Chi-Yuan}\binits{C.-Y.}} \betal{et~al.}
(\byear{2007}).
\btitle{Defining and estimating intervention effects for groups that will
  develop an auxiliary outcome}.
\bjournal{Statistical Science}
\bvolume{22}
\bpages{74--97}.
\end{barticle}
\endbibitem

\bibitem[\protect\citeauthoryear{Kitamura and
  Stutzer}{1997}]{kitamura1997information}
\begin{barticle}[author]
\bauthor{\bsnm{Kitamura},~\bfnm{Yuichi}\binits{Y.}} \AND
  \bauthor{\bsnm{Stutzer},~\bfnm{Michael}\binits{M.}}
(\byear{1997}).
\btitle{An information-theoretic alternative to generalized method of moments
  estimation}.
\bjournal{Econometrica}
\bvolume{65}
\bpages{861--874}.
\end{barticle}
\endbibitem

\bibitem[\protect\citeauthoryear{Lendle, Subbaraman and van~der
  Laan}{2013}]{lendle2013identification}
\begin{barticle}[author]
\bauthor{\bsnm{Lendle},~\bfnm{Samuel~D}\binits{S.~D.}},
  \bauthor{\bsnm{Subbaraman},~\bfnm{Meenakshi~S}\binits{M.~S.}} \AND
  \bauthor{\bparticle{van~der} \bsnm{Laan},~\bfnm{Mark~J}\binits{M.~J.}}
(\byear{2013}).
\btitle{Identification and efficient estimation of the natural direct effect
  among the untreated}.
\bjournal{Biometrics}
\bvolume{69}
\bpages{310--317}.
\end{barticle}
\endbibitem

\bibitem[\protect\citeauthoryear{MacKinnon}{2008}]{mackinnon2008introduction}
\begin{bbook}[author]
\bauthor{\bsnm{MacKinnon},~\bfnm{David~Peter}\binits{D.~P.}}
(\byear{2008}).
\btitle{Introduction to statistical mediation analysis}.
\bpublisher{Routledge}.
\end{bbook}
\endbibitem

\bibitem[\protect\citeauthoryear{Pearl}{2001}]{pearl2001direct}
\begin{barticle}[author]
\bauthor{\bsnm{Pearl},~\bfnm{Judea}\binits{J.}}
(\byear{2001}).
\btitle{Direct and Indirect Effects}.
\bjournal{Proceedings of the seventeenth conference on uncertainty in articial
  intelligence}
\bpages{411--420}.
\end{barticle}
\endbibitem

\bibitem[\protect\citeauthoryear{Pearl}{2012}]{pear:12}
\begin{barticle}[author]
\bauthor{\bsnm{Pearl},~\bfnm{Judea}\binits{J.}}
(\byear{2012}).
\btitle{The Causal Mediation Formula: A Guide to the Assessment of Pathways and
  Mechanisms}.
\bjournal{Prevention Science}
\bvolume{13}
\bpages{426--436}.
\end{barticle}
\endbibitem

\bibitem[\protect\citeauthoryear{Qin and Zhang}{2007}]{qin2007empirical}
\begin{barticle}[author]
\bauthor{\bsnm{Qin},~\bfnm{Jing}\binits{J.}} \AND
  \bauthor{\bsnm{Zhang},~\bfnm{Biao}\binits{B.}}
(\byear{2007}).
\btitle{Empirical-likelihood-based inference in missing response problems and
  its application in observational studies}.
\bjournal{Journal of the Royal Statistical Society: Series B (Statistical
  Methodology)}
\bvolume{69}
\bpages{101--122}.
\end{barticle}
\endbibitem

\bibitem[\protect\citeauthoryear{Richardson and Robins}{2013}]{rich:robi:13}
\begin{btechreport}[author]
\bauthor{\bsnm{Richardson},~\bfnm{Thomas~S.}\binits{T.~S.}} \AND
  \bauthor{\bsnm{Robins},~\bfnm{James~M.}\binits{J.~M.}}
(\byear{2013}).
\btitle{Single World Intervention Graphs ({SWIGs}): A Unification of the
  Counterfactual and Graphical Approaches to Causality}
\btype{Technical Report},
\bpublisher{University of Washington}.
\end{btechreport}
\endbibitem

\bibitem[\protect\citeauthoryear{Robins}{2003}]{robi:03}
\begin{bincollection}[author]
\bauthor{\bsnm{Robins},~\bfnm{James~M.}\binits{J.~M.}}
(\byear{2003}).
\btitle{Semantics of causal {DAG} models and the identification of direct and
  indirect effects}.
In \bbooktitle{Highly Structured Stochastic Systems (eds., {P.J. Green}, {N.L.
  Hjort}, and {S. Richardson})}
\bpages{70--81}.
\bpublisher{Oxford University Press}, \baddress{Oxford}.
\end{bincollection}
\endbibitem

\bibitem[\protect\citeauthoryear{Robins and
  Greenland}{1992}]{robins1992identifiability}
\begin{barticle}[author]
\bauthor{\bsnm{Robins},~\bfnm{James~M}\binits{J.~M.}} \AND
  \bauthor{\bsnm{Greenland},~\bfnm{Sander}\binits{S.}}
(\byear{1992}).
\btitle{Identifiability and exchangeability for direct and indirect effects}.
\bjournal{Epidemiology}
\bpages{143--155}.
\end{barticle}
\endbibitem

\bibitem[\protect\citeauthoryear{Robins, Rotnitzky and
  Zhao}{1994}]{robi:rotn:zhao:94}
\begin{barticle}[author]
\bauthor{\bsnm{Robins},~\bfnm{James~M.}\binits{J.~M.}},
  \bauthor{\bsnm{Rotnitzky},~\bfnm{Andrea}\binits{A.}} \AND
  \bauthor{\bsnm{Zhao},~\bfnm{Lue~Ping}\binits{L.~P.}}
(\byear{1994}).
\btitle{Estimation of Regression Coefficients When Some Regressors are Not
  Always Observed}.
\bjournal{Journal of the American Statistical Association}
\bvolume{89}
\bpages{846--866}.
\end{barticle}
\endbibitem

\bibitem[\protect\citeauthoryear{Sobel}{2008}]{sobe:08}
\begin{barticle}[author]
\bauthor{\bsnm{Sobel},~\bfnm{Michael~E.}\binits{M.~E.}}
(\byear{2008}).
\btitle{Identification of Causal Parameters in Randomized Studies with
  Mediating Variables}.
\bjournal{Journal of Educational and Behavioral Statistics}
\bvolume{33}
\bpages{230--251}.
\end{barticle}
\endbibitem

\bibitem[\protect\citeauthoryear{Tchetgen~Tchetgen and
  Shpitser}{2012}]{tchetgen2012semiparametric}
\begin{barticle}[author]
\bauthor{\bsnm{Tchetgen~Tchetgen},~\bfnm{Eric~J}\binits{E.~J.}} \AND
  \bauthor{\bsnm{Shpitser},~\bfnm{Ilya}\binits{I.}}
(\byear{2012}).
\btitle{Semiparametric theory for causal mediation analysis: efficiency bounds,
  multiple robustness and sensitivity analysis}.
\bjournal{The Annals of Statistics}
\bvolume{40}
\bpages{1816--1845}.
\end{barticle}
\endbibitem

\bibitem[\protect\citeauthoryear{Ten~Have et~al.}{2007}]{ten2007causal}
\begin{barticle}[author]
\bauthor{\bsnm{Ten~Have},~\bfnm{Thomas~R}\binits{T.~R.}},
  \bauthor{\bsnm{Joffe},~\bfnm{Marshall~M}\binits{M.~M.}},
  \bauthor{\bsnm{Lynch},~\bfnm{Kevin~G}\binits{K.~G.}},
  \bauthor{\bsnm{Brown},~\bfnm{Gregory~K}\binits{G.~K.}},
  \bauthor{\bsnm{Maisto},~\bfnm{Stephen~A}\binits{S.~A.}} \AND
  \bauthor{\bsnm{Beck},~\bfnm{Aaron~T}\binits{A.~T.}}
(\byear{2007}).
\btitle{Causal mediation analyses with rank preserving models}.
\bjournal{Biometrics}
\bvolume{63}
\bpages{926--934}.
\end{barticle}
\endbibitem

\bibitem[\protect\citeauthoryear{Tseng and
  Bertsekas}{1987}]{tseng1987relaxation}
\begin{barticle}[author]
\bauthor{\bsnm{Tseng},~\bfnm{Paul}\binits{P.}} \AND
  \bauthor{\bsnm{Bertsekas},~\bfnm{Dimitri~P}\binits{D.~P.}}
(\byear{1987}).
\btitle{Relaxation methods for problems with strictly convex separable costs
  and linear constraints}.
\bjournal{Mathematical Programming}
\bvolume{38}
\bpages{303--321}.
\end{barticle}
\endbibitem

\bibitem[\protect\citeauthoryear{VanderWeele}{2009}]{vanderweele2009marginal}
\begin{barticle}[author]
\bauthor{\bsnm{VanderWeele},~\bfnm{Tyler~J}\binits{T.~J.}}
(\byear{2009}).
\btitle{Marginal structural models for the estimation of direct and indirect
  effects}.
\bjournal{Epidemiology}
\bvolume{20}
\bpages{18--26}.
\end{barticle}
\endbibitem

\bibitem[\protect\citeauthoryear{VanderWeele}{2010}]{vand:10}
\begin{barticle}[author]
\bauthor{\bsnm{VanderWeele},~\bfnm{Tyler~J.}\binits{T.~J.}}
(\byear{2010}).
\btitle{Bias formulas for sensitivity analysis for direct and indirect
  effects}.
\bjournal{Epidemiology}
\bvolume{21}
\bpages{540--551}.
\end{barticle}
\endbibitem

\bibitem[\protect\citeauthoryear{VanderWeele}{2015}]{vand:15}
\begin{bbook}[author]
\bauthor{\bsnm{VanderWeele},~\bfnm{Tyler}\binits{T.}}
(\byear{2015}).
\btitle{Explanation in Causal Inference: Methods for Mediation and
  Interaction}.
\bpublisher{Oxford University Press}.
\end{bbook}
\endbibitem

\bibitem[\protect\citeauthoryear{VanderWeele and
  Vansteelandt}{2010}]{vanderweele2010odds}
\begin{barticle}[author]
\bauthor{\bsnm{VanderWeele},~\bfnm{Tyler~J}\binits{T.~J.}} \AND
  \bauthor{\bsnm{Vansteelandt},~\bfnm{Stijn}\binits{S.}}
(\byear{2010}).
\btitle{Odds ratios for mediation analysis for a dichotomous outcome}.
\bjournal{American journal of epidemiology}
\bvolume{172}
\bpages{1339--1348}.
\end{barticle}
\endbibitem

\bibitem[\protect\citeauthoryear{VanderWeele, Vansteelandt and
  Robins}{2014}]{vanderweele2014effect}
\begin{barticle}[author]
\bauthor{\bsnm{VanderWeele},~\bfnm{Tyler~J}\binits{T.~J.}},
  \bauthor{\bsnm{Vansteelandt},~\bfnm{Stijn}\binits{S.}} \AND
  \bauthor{\bsnm{Robins},~\bfnm{James~M}\binits{J.~M.}}
(\byear{2014}).
\btitle{Effect decomposition in the presence of an exposure-induced
  mediator-outcome confounder}.
\bjournal{Epidemiology}
\bvolume{25}
\bpages{300--306}.
\end{barticle}
\endbibitem

\bibitem[\protect\citeauthoryear{Vinokur and Schul}{1997}]{vinokur1997mastery}
\begin{barticle}[author]
\bauthor{\bsnm{Vinokur},~\bfnm{Amiram~D}\binits{A.~D.}} \AND
  \bauthor{\bsnm{Schul},~\bfnm{Yaacov}\binits{Y.}}
(\byear{1997}).
\btitle{Mastery and inoculation against setbacks as active ingredients in the
  JOBS intervention for the unemployed.}
\bjournal{Journal of consulting and clinical psychology}
\bvolume{65}
\bpages{867--877}.
\end{barticle}
\endbibitem

\end{thebibliography}


\end{document}